\documentclass[a4paper,UKenglish,final]{lipics-v2019}

\title{Monte Carlo Tree Search guided by Symbolic Advice for MDPs}

\author{Damien Busatto-Gaston}{Université Libre de Bruxelles, Brussels, Belgium }{damien.busatto-gaston@ulb.ac.be}{https://orcid.org/0000-0002-7266-0927}{}%

\author{Debraj Chakraborty}{Université Libre de Bruxelles, Brussels, Belgium }{debraj.chakraborty@ulb.ac.be}{https://orcid.org/0000-0003-0978-4457}{}

\author{Jean-Francois Raskin}{Université Libre de Bruxelles, Brussels, Belgium }{jraskin@ulb.ac.be}{}{}

\authorrunning{D. Busatto-Gaston, D. Chakraborty and J-F. Raskin} %

\Copyright{Damien Busatto-Gaston, Debraj Chakraborty and Jean-Francois Raskin} %



\funding{This work is partially supported by the ARC project Non-Zero Sum Game Graphs: Applications to Reactive Synthesis and Beyond (F\'{e}d\'{e}ration Wallonie-Bruxelles), the EOS project Verifying Learning Artificial Intelligence Systems (F.R.S.-FNRS \& FWO), and the COST Action 16228 GAMENET (European Cooperation in Science and Technology). }%

\acknowledgements{Computational resources have been provided by the Consortium des Équipements de Calcul Intensif (CÉCI), funded by the Fonds de la Recherche Scientifique de Belgique (F.R.S.-FNRS) under Grant No. 2.5020.11.}%

\nolinenumbers %
\hideLIPIcs  %

\EventEditors{Igor Konnov and Laura Kov\'{a}cs}
\EventNoEds{2}
\EventLongTitle{31st International Conference on Concurrency Theory (CONCUR 2020)}
\EventShortTitle{CONCUR 2020}
\EventAcronym{CONCUR}
\EventYear{2020}
\EventDate{September 1--4, 2020}
\EventLocation{Vienna, Austria}
\EventLogo{}
\SeriesVolume{2017}
\ArticleNo{30}

\usepackage[obeyFinal]{todonotes}
\usepackage{showkeys}

\usepackage{amsfonts,amssymb,amsmath}

\usepackage{url}

\usepackage{hyperref}

\usepackage{xspace} %
\usepackage{xparse} %

\newcommand\newmath[2]{\newcommand#1{\ensuremath{#2}\xspace}}
\newcommand\renewmath[2]{\renewcommand#1{\ensuremath{#2}\xspace}}
\newcommand\newmathope[2]{\newcommand#1{\ensuremath{\operatornamewithlimits{#2}}\xspace}}

\newmath{\N}{\mathbb{N}}
\newmath{\Z}{\mathbb{Z}}
\newmath{\Q}{\mathbb{Q}}
\newmath{\R}{\mathbb{R}}

\newmathope{\argmin}{\arg\min}
\newmathope{\argmax}{\arg\max}

\renewmath{\Pr}{\mathbb P}
\newmath{\Dist}{\mathcal D}
\newmath{\Supp}{\mathsf{Supp}}
\newmath{\E}{\mathbb E}

\newmath{\Reward}{\mathsf{Reward}}
\newmath{\AReward}{\mathsf{AReward}}

\newmath{\FPaths}{\mathsf{Paths}}
\newmath{\first}{\mathsf{first}}
\newmath{\last}{\mathsf{last}}
\newmath{\Val}{\mathsf{Val}}

\newmath{\opt}{\mathsf{opt}}

\newmath{\reward}{\mathsf{reward}}
\newmath{\numsamples}{\mathsf{count}} %
\newmath{\children}{\mathsf{children}}
\newmath{\mctsvalue}{\mathsf{value}}
\newmath{\total}{\mathsf{total}}
\newmath{\I}{\mathcal{I}}
\newmath{\iter}{\mathsf{iter}}

\newmath{\A}{\mathscr A}
\newmath{\Apsi}{\A^{\psi}}
\newmath{\Aone}{\Apsi_1}
\newmath{\Aphi}{\A^{\varphi}}
\newmath{\Atwo}{\Aphi_2}

\renewcommand\paragraph[1]{\smallskip\noindent\textbf{#1.}}

\begin{document}

\maketitle

\begin{abstract}
In this paper, we consider the online computation of a strategy that aims
at optimizing the expected average reward in a Markov decision process.
The strategy is computed with a receding horizon and using
Monte Carlo tree search (MCTS). We augment the MCTS algorithm with
the notion of symbolic advice, and show that its classical
theoretical guarantees are maintained. Symbolic advice
are used to bias the selection and simulation
strategies of MCTS.
We describe how to use QBF and SAT solvers to implement symbolic advice
in an efficient way. We illustrate our new algorithm using
the popular game {\sc{Pac-Man}} and show that the performances
of our algorithm exceed those of plain MCTS as well as
the performances of human players.
\end{abstract}

\section{Introduction}\label{sec:intro}

Markov decision processes (MDP) are an important mathematical formalism for modeling and solving sequential decision problems in stochastic environments~\cite{DBLP:books/wi/Puterman94}. The importance of this model has triggered a large number of works in different research communities within computer science, most notably in formal verification, and in artificial intelligence and machine learning. The works done in these research communities have respective weaknesses and complementary strengths. On the one hand, algorithms developed in formal verification are generally complete and provide strong guarantees on the optimality of computed solutions but they tend to be applicable to models of moderate size only. On the other hand, algorithms developed in artificial intelligence and machine learning usually scale to larger models but only provide weaker guarantees.  Instead of opposing the two sets of algorithms, there have been recent works~\cite{DBLP:conf/isola/AshokBKS18,DBLP:conf/atal/HasanbeigAK20,DBLP:conf/atva/BrazdilCCFKKPU14,DBLP:conf/tacas/DacaHKP16,DBLP:conf/aaai/Chatterjee0PRZ17,DBLP:conf/concur/KretinskyPR18,DBLP:conf/aaai/AlshiekhBEKNT18} that try to combine the strengths of the two approaches in order to offer new hybrid algorithms that %
scale better and %
provide stronger guarantees. The contributions described in this paper are part of this research agenda: we show how to integrate {\em symbolic advice} defined by formal specifications into Monte Carlo Tree Search algorithms~\cite{DBLP:journals/tciaig/BrownePWLCRTPSC12} using techniques such as SAT~\cite{DBLP:reference/mc/Marques-SilvaM18} and QBF~\cite{DBLP:conf/ictai/ShuklaBPS19}.

When an MDP is too large to be analyzed {\em offline} using verification algorithms, receding horizon analysis combined with simulation techniques are used {\em online}~\cite{DBLP:journals/ml/KearnsMN02}. Receding horizon techniques work as follows. In the current state $s$ of the MDP, for a fixed horizon $H$, the receding horizon algorithm searches for an action $a$ that is the first action of a plan to act (almost) optimally on the finite horizon $H$. When such an action is identified, then it is played from $s$ and the state evolves stochastically to a new state $s'$ according to the dynamics specified by the MDP. The same process is repeated from $s'$. The optimization criterion over the~$H$ next step depends on the long run measure that needs to be optimised. The tree unfolding from $s$ that needs to be analyzed is often very large (\textit{e.g.}~it may be exponential in $H$). As a consequence, receding horizon techniques are often coupled with {\em sampling techniques} that avoid the systematic exploration of the entire tree unfolding at the expense of approximation. The Monte Carlo Tree Search (MCTS) algorithm \cite{DBLP:journals/tciaig/BrownePWLCRTPSC12} is an increasingly popular tree search algorithm that implements these ideas. It is one of the core building blocks of the {\sc AlphaGo} algorithm~\cite{DBLP:journals/nature/SilverHMGSDSAPL16}.

While MCTS techniques may offer reasonable performances out of the shelf, they usually need substantial adjustments that depend on the application to really perform well. One way to adapt MCTS to a particular application is to {\em bias} the search towards promising subspaces taking into account properties of the application domain~\cite{5695450,silver2009monte}. This is usually done by coding directly handcrafted search and sampling strategies. We show in this paper how to use techniques from formal verification to offer a {\em flexible} and {\em rigorous} framework to bias the search performed by MCTS using {\em symbolic advice}. A symbolic advice is a {\em formal specification}, that can be expressed for example in your favorite linear temporal logic, and which constrain the search and the sampling phases of the MCTS algorithm using QBF and SAT solvers. Our framework offers in principle the ability to easily experiment with precisely formulated bias expressed declaratively using logic.

\paragraph{Contributions} On the theoretical side, we study the impact of using symbolic advice on the guarantees offered by MCTS. We identify sufficient conditions for the symbolic advice to preserve the convergence guarantees of the MCTS algorithm (Theorem~\ref{thm:advMCTS}). Those results are partly based on an analysis of the incidence of sampling on those guarantees (Theorem~\ref{thm:MCTS}) which can be of independent interest.

On a more practical side, we show how symbolic advice can be implemented using SAT and QBF techniques. More precisely, we use QBF~\cite{DBLP:conf/cav/NarodytskaLBRW14} to force that all the prefixes explored by the MCTS algorithm in the partial tree unfolding have the property suggested by the {\em selection advice} (whenever possible) and we use SAT-based sampling techniques~\cite{DBLP:conf/tacas/ChakrabortyFMSV15} to achieve uniform sampling among paths of the MDP that satisfy the {\em sampling advice}. The use of this symbolic exploration techniques is important as the underlying state space that we need to analyze is usually huge (\textit{e.g.}~exponential in the receding horizon $H$).

\begin{figure}
\begin{subfigure}{.45\textwidth}
  \centering
  \includegraphics[height=0.19\textheight,keepaspectratio]{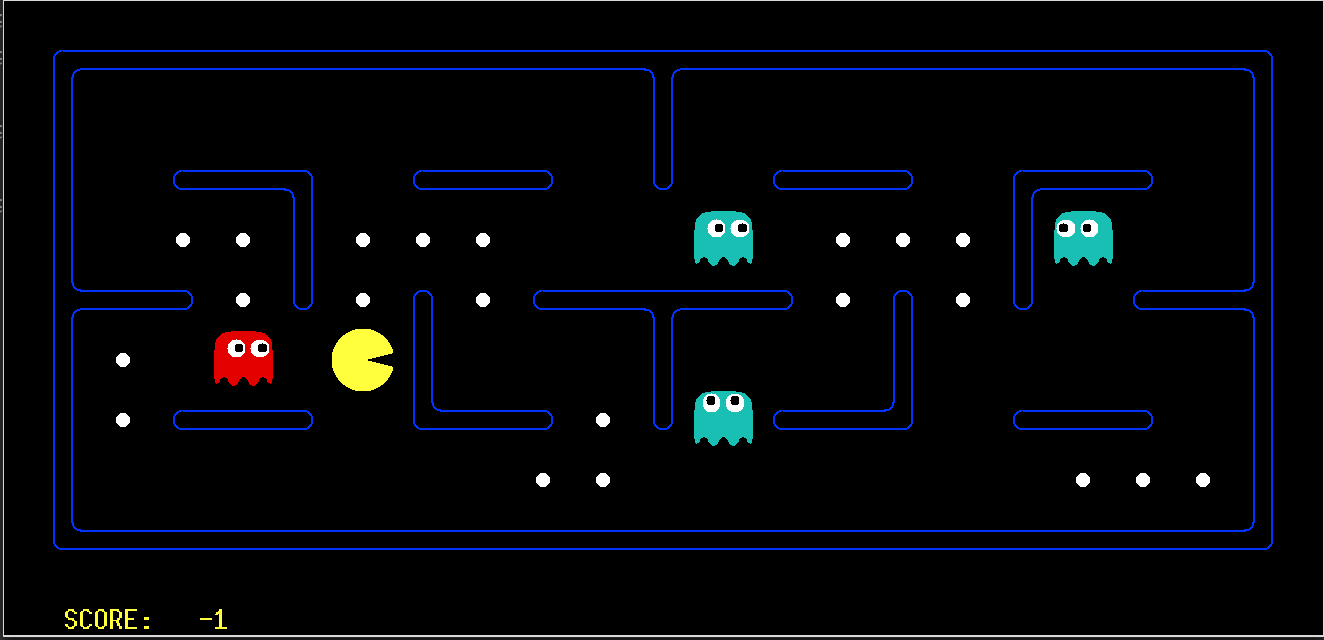}
\end{subfigure}
\begin{subfigure}{.75\textwidth}
  \centering
  \includegraphics[height=0.19\textheight,keepaspectratio]{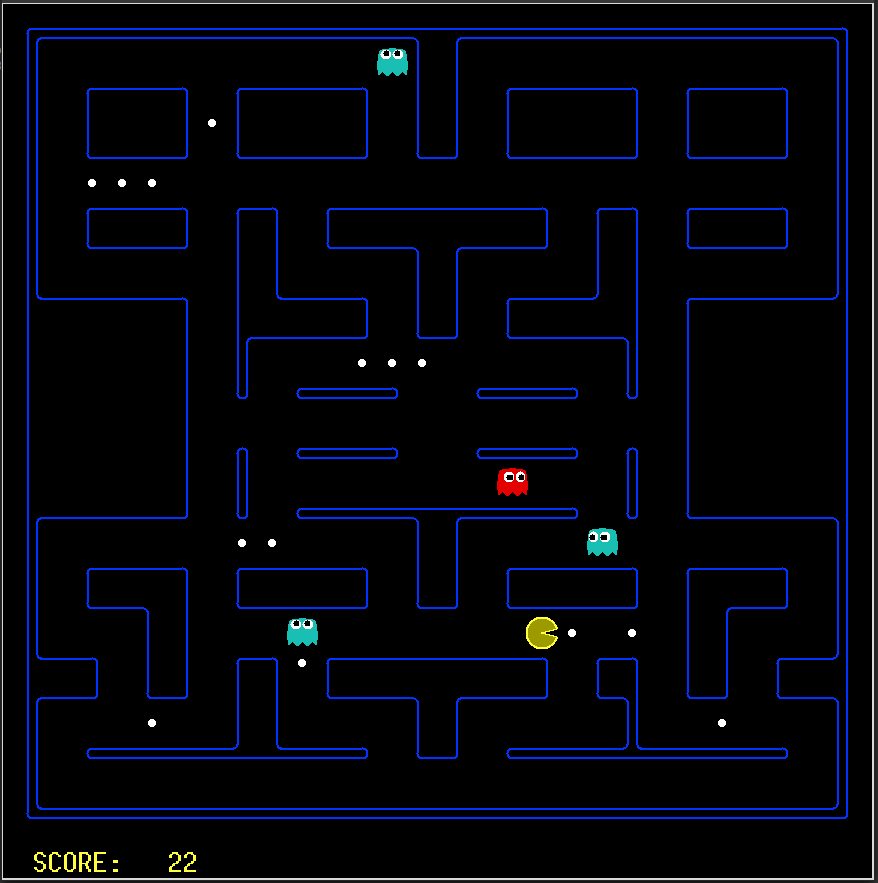}
\end{subfigure}%
\caption{We used two grids of size $9\times 21$ and $27\times 28$ %
for our experiments.
Pac-Man loses if he makes contact with a ghost, and wins if he eats all food pills (in white).
The agents can travel in four directions unless they are blocked by the walls in the grid, and ghosts cannot reverse their direction. The score decreases by $1$ at each step, and increases by $10$ whenever Pac-Man eats a food pill.
A win (resp. loss), increases (resp. decreases) the score by $500$.
The game can be seen as an infinite duration game by saying that whenever Pac-Man wins or loses,
the positions of the agents and of the food pills are reset.
}%
\label{fig:grids}
\end{figure}

To demonstrate the practical interest of our techniques, we have applied our new {\em MCTS with symbolic advice algorithm} to play {\sc{Pac-Man}}. Figure~\ref{fig:grids} shows a grid of the {\sc{Pac-Man}} game. In this version of the classical game, the agent Pac-Man has to eat food pills as fast as possible while avoiding being pinched by ghosts. We have chosen this benchmark to evaluate our algorithm for several reasons. First, the state space of the underlying MDP is way too large for the state-of-the-art implementations of complete algorithms. Indeed, the reachable state space of the small grid shown here has approximately $10^{16}$ states, while the classical grid has approximately $10^{23}$ states. Our algorithm can handle both grids. Second, this application not only allows for comparison between performances obtained from several versions of the MCTS algorithm but also with the performances that humans can attain in this game. In the {\sc{Pac-Man}} benchmark, we show that advice that instructs Pac-Man on the one hand to avoid ghosts at all costs during the selection phase of the MCTS algorithm (enforced whenever possible by QBF) and on the other hand to bias the search to paths in which ghosts are avoided (using uniform sampling based on SAT) allow to attain or surpass human level performances while the standard MCTS algorithm performs much worse.

\paragraph{Related works}
Our analysis of the convergence of the MCTS algorithm with appropriate symbolic advice is based on extensions of analysis results based on bias defined using UCT (bandit algorithms)~\cite{DBLP:conf/ecml/KocsisS06,DBLP:journals/ml/AuerCF02}. Those results are also related to sampling techniques for finite horizon objectives in MDP~\cite{DBLP:journals/ml/KearnsMN02}.

Our concept of selection phase advice is related to the MCTS-minimax hybrid algorithm proposed in~\cite{DBLP:journals/tciaig/BaierW15}. There the selection phase advice is not specified declaratively using logic but encoded directly in the code of the search strategy. No use of QBF nor SAT is advocated there and no use of sampling advice either.
In~\cite{DBLP:conf/aaai/AlshiekhBEKNT18}, the authors provide a general framework to add safety properties to reinforcement learning algorithms via {\em shielding}.
These techniques analyse statically the full state space of the game in order to compute a set of unsafe actions
to avoid. This fits our advice framework, so that such a shield could be used as an online selection advice
in order to combine their safety guarantees with our formal results for MCTS.
More recently, a variation of shielding called {\em safe padding} has been studied in~\cite{DBLP:conf/atal/HasanbeigAK20}. Both works are concerned with reinforcement learning and not with MCTS.
Note that in general multiple ghosts may prevent the existence of a strategy to enforce safety,
\textit{i.e.}~always avoid pincer moves.

Our practical handling of symbolic sampling advice relies on symbolic sampling techniques introduced in~\cite{DBLP:conf/aaai/ChakrabortyFMSV14}, while our handling of symbolic selection advice relies on natural encodings via QBF that are similar to those defined in~\cite{DBLP:conf/cav/NarodytskaLBRW14}.

\section{Preliminaries}\label{sec:prelim}

A \emph{probability distribution} on a finite set $S$ is a function $d:S\to [0,1]$ such that $\sum_{s\in S}d(s)=1$.
We denote the set of all probability distributions on set $S$ by $\Dist(S)$. The support of a distribution $d\in \Dist(S)$ is $\Supp(d)=\{s\in S\mid d(s)>0\}$.

\subsection{Markov decision process}

\begin{definition}[MDP]\label{def:mdp}
A Markov decision process is a tuple $M=(S,A,P,R,R_T)$, where
$S$ is a finite set of states, $A$ is a finite set of actions,
$P$ is a mapping from $S\times A$ to $\Dist(S)$ such that $P(s,a)(s')$
denotes the probability that action $a$ in state $s$ leads to state $s'$,
$R:S\times A\to \R$ defines the reward obtained for taking a given action
at a given state, and $R_T:S\to \R$ assigns a terminal reward to each state in $S$.\footnote{
We assume for convenience that every action in $A$ can be taken from every state.
One may need to limit this choice to a subset of \emph{legal} actions that depends on the current state.
This concept can be encoded in our formalism by adding a sink state reached with probability $1$
when taking an illegal action.}
\end{definition}

For a Markov decision process $M$, a \emph{path} of length $i>0$ is
a sequence of $i$ consecutive states and actions followed by a last state.
We say that $p = s_0a_0s_1\ldots s_i$ is an $i$-length path in the MDP $M$
if for all $t\in[0,i-1]$, $a_t\in A$ and $s_{t+1}\in \Supp(P(s_t,a_t))$,
and we denote $\last(p) = s_i$ and $\first(p) = s_0$.
We also consider states to be paths of length $0$.
An infinite path is an infinite sequence $p = s_0a_0s_1\ldots$ of states and actions
such that for all $t\in\N$, $a_t\in A$ and $s_{t+1}\in \Supp(P(s_t,a_t))$.
We denote the finite prefix of length $t$ of a finite or infinite path $p=s_0a_0s_1\ldots$ by $p_{|t}=s_0a_0\ldots s_t$.
Let $p=s_0a_0s_1\ldots s_i$ and $p'=s'_0a'_0s'_1\ldots s'_j$ be two paths
such that $s_i=s'_0$, let $a$ be an action and $s$ be state of $M$.
Then, $p\cdot p'$
denotes $s_0a_0s_1\ldots s_ia'_0s'_1\ldots s'_j$ and $p\cdot as$ denotes $s_0a_0s_1\ldots s_ias$.

For an MDP $M$, the set of all finite paths of length $i$ %
is denoted by $\FPaths^i_{M}$. %
Let $\FPaths^i_{M}(s)$ denote the set of paths $p$ in $\FPaths^i_{M}$
such that $\first(p)=s$.
Similarly, if $p\in\FPaths^i_{M}$ and $i\leq j$, then let $\FPaths^j_{M}(p)$ denote the set of paths $p'$ in $\FPaths^{j}_{M}$ such that there exists $p''\in\FPaths^{j-i}_M$ with $p'=p\cdot p''$.
We denote the set of all finite paths in $M$ by $\FPaths_{M}$ and the set of finite paths of length
at most $H$ by $\FPaths^{\leq H}_{M}$.

\begin{definition}%
The total reward of a finite path $p = s_0a_0\dots s_n$ in $M$ %
is defined as $$\Reward_{M}(p)=\sum_{t=0}^{n-1}R(s_t,a_t)+R_T(s_n)\,.$$
\end{definition}

A (probabilistic) \emph{strategy} is a function $\sigma : \FPaths_{M} \to \Dist(A)$
that maps a path $p$ to a probability distribution in $\Dist(A)$.
A strategy $\sigma$ is \emph{deterministic} if the support of
the probability distributions $\sigma(p)$ has size $1$,
it is \emph{memoryless} if $\sigma(p)$ depends only on $\last(p)$,
\textit{i.e.}~if $\sigma$ satisfies that for all $p,p'\in\FPaths_{M}$, $\last(p)=\last(p')\Rightarrow\sigma(p)=\sigma(p')$.
For a probabilistic strategy $\sigma$ and $i\in\N$, let $\FPaths^i_{M}(\sigma)$
denote the paths $p = s_0a_0\ldots s_i$ in $\FPaths^i_{M}$
such that for all $t\in[0,i-1]$, $a_t\in \Supp(\sigma(p_{|t}))$.
For a finite path $p$ of length $i\in\N$ and some $j\geq i$,
let $\FPaths^j_{M}(p,\sigma)$ denote $\FPaths^j_{M}(\sigma)\cap\FPaths^j_{M}(p)$.

For a strategy $\sigma$ and a path $p\in\FPaths^i_{M}(\sigma)$,
let the probability of $p=s_0a_0\dots s_i$
in $M$ according to $\sigma$ be defined as $\Pr_{M,\sigma}^i(p)=
\prod_{t=0}^{i-1} \sigma(p_{|t})(a_t) P(s_t,a_t)(s_{t+1})$.
The mapping $\Pr_{M,\sigma}^i$ defines a probability distribution
over $\FPaths^i_{M}(\sigma)$.

\begin{definition}%
The expected average reward of a probabilistic strategy $\sigma$ in an MDP $M$,
starting from state $s$, is defined as
$$\Val_{M}(s,\sigma) = \liminf_{n\to \infty}\cfrac{1}{n}\E\left[\Reward_{M}(p)\right]\,,$$
where $p$ is a random variable over $\FPaths^n_{M}(\sigma)$ following
the distribution $\Pr_{M,\sigma}^n$.
\end{definition}

\begin{definition}%
The optimal expected average reward starting from a state $s$ in an MDP $M$%
is defined over all strategies $\sigma$ in $M$ as
$\Val_{M}(s)=\sup_{\sigma}\Val_{M}(s,\sigma)$.
\end{definition}

One can restrict the supremum to deterministic memoryless
strategies~\cite[Proposition 6.2.1]{DBLP:books/wi/Puterman94}.
A strategy $\sigma$ is called $\epsilon$-optimal for the expected average reward
if $\Val_M(s,\sigma)\geq \Val_M(s)-\epsilon$ for all $s$.

\begin{definition}%
  The expected total reward of a probabilistic strategy $\sigma$ in an MDP $M$,
  starting from state $s$ and for a finite horizon $i$, is defined as
  $\Val_{M}^i(s,\sigma) = \E\left[\Reward_{M}(p)\right]$,
  where $p$ is a random variable over $\FPaths^i_{M}(\sigma)$ following the distribution $\Pr_{M,\sigma}^i$.
\end{definition}

\begin{definition}%
The optimal expected total reward starting from a state $s$ in an MDP $M$, with horizon $i\in\N$,
is defined over all strategies $\sigma$ in $M$ as
$\Val_{M}^i(s) = \sup_{\sigma}\Val_{M}^i(s,\sigma)$.
\end{definition}
One can restrict the supremum to deterministic
strategies~\cite[Theorem 4.4.1.b]{DBLP:books/wi/Puterman94}.

Let $\sigma^{i,*}_{M,s}$ denote a deterministic strategy that
maximises $\Val_M^i(s,\sigma)$, and refer to it as
an optimal strategy for the expected total reward of horizon $i$ at state $s$.
For $i\in \N$, let $\sigma^{i}_{M}$ refer to a deterministic memoryless strategy
that maps every state $s$ in $M$ to the first action of a corresponding optimal strategy for
the expected total reward of horizon $i$, so that $\sigma^{i}_{M}(s)=\sigma^{i,*}_{M,s}(s)$.
As there may exist several deterministic strategies $\sigma$ that
maximise $\Val_M^i(s,\sigma)$, we denote by $\opt^{i}_{M}(s)$
the set of actions $a$ such that there exists an optimal strategy $\sigma^{i,*}_{M,s}$
that selects $a$ from $s$.
A strategy $\sigma^{i}_{M}$ can be obtained by the value iteration algorithm:

\begin{proposition}[{Value iteration~\cite[Section 5.4]{DBLP:books/wi/Puterman94}}]\label{prop:ValueIte}
  For a state $s$ in MDP $M$, for all $i\in\N$,
  \begin{itemize}
    \item $\Val^{i+1}_{M}(s)=\max_{a\in A}\left[R(s,a)+\sum_{s'}P(s,a)(s')\Val^i_{M}(s')\right]$
    \item $\opt^{i+1}_{M}(s)=\argmax_{a\in A}\left[R(s,a)+\sum_{s'}P(s,a)(s')\Val^i_{M}(s')\right]$
  \end{itemize}
\end{proposition}

Moreover, for a large class of MDPs and a large enough $n$, the strategy $\sigma^{n}_{M}$ is $\epsilon$-optimal for the expected average reward: %
\begin{proposition}\cite[Theorem 9.4.5]{DBLP:books/wi/Puterman94}\label{prop:ValItCon}
  For a strongly aperiodic\footnote{A Markov decision process is strongly aperiodic if $P(s,a)(s)>0$
  for all $s\in S$ and $a\in A$.} Markov decision process $M$, %
  it holds that $\Val_M(s)=\lim_{n\to\infty}[\Val^{n+1}_M(s)-\Val^n_M(s)]$.
  Moreover, for any $\epsilon>0$ there exists $N\in\N$ such that for all $n\geq N$, $Val_M(s,\sigma^n_M)\geq Val_M(s)-\epsilon$ for all $s$.
\end{proposition}

A simple transformation can be used to make an MDP strongly aperiodic
without changing the optimal expected average reward and the associated optimal strategies.
Therefore, one can use an algorithm computing the strategy $\sigma^H_M$
in order to optimise for the expected average reward, and obtain theoretical
guarantees for a horizon $H$ big enough. This is known as the \emph{receding horizon} approach.

Finally, we will use the notation $T(M,s_0,H)$ to refer to an MDP obtained as a
tree-shaped \emph{unfolding} of $M$ from state $s_0$ and for a depth of $H$.
In particular, the states of $T(M,s_0,H)$ correspond to paths in $\FPaths^{\leq H}_{M}(s_0)$.
Then, it holds that:
\begin{lemma}\label{lm:unfolding}
  $\Val^{H}_{M}(s_0)$ is equal to $\Val^{H}_{T(M,s_0,H)}(s_0)$,
  and $\opt^H_M(s_0)$ is equal to $\opt^H_{T(M,s_0,H)}(s_0)$.
\end{lemma}

The aperiodicity and unfolding transformations are detailed in Appendix~\ref{app:mdp}.

\subsection{Bandit problems and UCB}\label{sec:UCB}

In this section, we present bandit problems, whose study
forms the basis of a theoretical analysis of Monte Carlo tree search algorithms.

Let $A$ denote a finite set of actions.
For each $a\in A$, let $(x_{a,t})_{t\geq 1}$
be a sequence of random payoffs associated to $a$.
They correspond to successive plays of action $a$, and for
every action $a$ and every $t\geq 1$, let $x_{a,t}$ be drawn
with respect to a probability distribution $\Dist_{a,t}$ over $[0,1]$.
We denote by $X_{a,t}$ the random variable associated to this drawing.
In a \emph{fixed distributions} setting (the classical bandit problem),
every action is associated to a fixed probability distribution $\Dist_a$,
so that $\Dist_{a,t}=\Dist_a$ for all $t\geq 1$.

The \emph{bandit problem} consists of a succession of steps where
the player selects an action and observes the associated payoff,
while trying to maximise the cumulative gains.
For example, selecting action $a$, then $b$ and then $a$ again
would yield the respective payoffs $x_{a,1}$, $x_{b,1}$ and $x_{a,2}$
for the first three steps, drawn from their respective distributions.
Let the regret $R_n$ denote the difference,
after $n$ steps, between the optimal
expected payoff $\max_{a\in A} \E[\sum_{t=1}^n X_{a,t}]$ and the expected payoff
associated to our action selection.
The goal is to minimise the long-term regret when the number of steps $n$ increases.

The algorithm UCB1 of \cite{DBLP:journals/ml/AuerCF02} offers a practical solution to this problem,
and offers theoretical guarantees. %
For an action $a$ and $n\geq 1$, let $\overline x_{a,n}=\frac{1}{n}\sum_{t=1}^{n}x_{a,t}$ denote
the average payoff obtained from the first $n$ plays of $a$.
Moreover, for a given step number $t$ let $t_a$ denote how many times action $a$
was selected in the first $t$ steps.
The algorithm UCB1 chooses, at step $t+1$, the action $a$ that maximises
  $\overline x_{a,t_a}+c_{t,t_a}$,
where
$c_{t,t_a}$ is defined as $\sqrt{\frac{2\ln t}{t_a}}$.
This procedure enjoys optimality guarantees detailed in \cite{DBLP:journals/ml/AuerCF02},
as it keeps the regret $R_n$ below $O(\log n)$.

We will make use of an extension of these results to the general setting
of \emph{non-stationary} bandit problems,
where the distributions $\Dist_{a,t}$ are no longer fixed with respect to $t$.
This problem has been studied in \cite{DBLP:conf/ecml/KocsisS06}, and results were obtained for
a class of distributions $\Dist_{a,t}$ that respect assumptions
referred to as \emph{drift conditions}.

For a fixed $n\geq 1$, let $\overline X_{a,n}$ denote the random variable
obtained as the average of the random variables associated with the first $n$
plays of $a$.
Let $\mu_{a,n}=\E[\overline X_{a,n}]$.
We assume that these expected means eventually converge,
and let $\mu_a=\lim_{n\to\infty}\mu_{a,n}$.

\begin{definition}[Drift conditions]\label{def:drift}
For all $a\in A$, the sequence $(\mu_{a,n})_{n\geq 1}$ converges to some value $\mu_a$.
Moreover, there exists a constant $C_p>0$ and an integer $N_p$ such that for $n\geq N_p$ and any $\delta>0$, if $\Delta_n(\delta)=C_p\sqrt{n\ln(1/\delta)}$ then
the tail inequalities
$\Pr[n\overline X_{a,n}\geq n \mu_{a,n} +\Delta_n(\delta)]\leq \delta$
and $\Pr[n\overline X_{a,n}\leq n \mu_{a,n} -\Delta_n(\delta)]\leq \delta$ hold.
\end{definition}

We recall in Appendix~\ref{app:UCB} the results of \cite{DBLP:conf/ecml/KocsisS06},
and provide an informal description of those results here.
Consider using the algorithm UCB1 on a non-stationary bandit problem satisfying the drift conditions,
with $c_{t,t_a}=2C_p\sqrt{\frac{\ln t}{t_a}}$.
First, one can bound logarithmically the number of times a suboptimal action is played.
This is used to bound the difference between $\mu_a$ and $\E[\overline X_{n}]$ by $O(\ln n/n)$,
where $a$ is an optimal action and where $\overline X_{n}$ denotes
the global average of payoffs received over the first $n$ steps.
This is the main theoretical guarantee obtained for the optimality of UCB1.
Also for any action $a$, the authors state a lower bound for the number of times the action is played.
The authors also prove a tail inequality similar to the one described
in the drift conditions, but on the random variable $\overline X_{n}$
instead of $\overline X_{a,n}$. This will be useful for inductive proofs later on,
when the usage of UCB1 is nested so that the global sequence $\overline X_{n}$
corresponds to a sequence $\overline X_{b,n}$ of an action $b$ played from the next state of the MDP. %
Finally, it is shown that the probability of making the wrong decision
(choosing a suboptimal action) converges to $0$ as the number of plays $n$ grows large enough.

\section{Monte Carlo tree search with simulation}\label{sec:MCTSSimulation}
In a receding horizon approach, the objective is to compute $\Val_M^H(s_0)$ and $\sigma^H_M$ for some
state $s_0$ and some horizon $H$. Exact procedures such as the recursive computation
of Proposition~\ref{prop:ValueIte} can not be used on large MDPs, resulting in heuristic approaches.
We focus on the Monte Carlo Tree Search (MCTS)
algorithm~\cite{DBLP:journals/tciaig/BrownePWLCRTPSC12},
that can be seen as computing approximations of $\Val^H_{M}$ and $\sigma^H_{M}(s_0)$
on the unfolding $T(M,s_0,H)$.
Note that rewards in the MDP $M$ are bounded.\footnote{There are
finitely many paths of length at most $H$, with rewards in $\R$.}
For the sake of simplicity we assume without loss of generality that for all paths $p$
of length at most $H$ the total reward $\Reward_M(p)$ belongs to $[0,1]$.

Given an initial state $s_0$, MCTS is an iterative process that incrementally
constructs a search tree rooted at $s_0$
describing paths of $M$ and their associated values.
This process goes on until a specified budget
(of number of iterations or time) is exhausted.
An iteration constructs a path in $M$ by following a decision strategy
to \emph{select} a sequence of nodes in the search tree. When a node that is not
part of the current search tree is reached, the tree is expanded with this new node,
whose expected reward is approximated by \emph{simulation}.
This value is then used to update the knowledge of all selected nodes in \emph{backpropagation}.

In the search tree, each node represents a path.
For a node $p$ and an action $a\in A$, let $\children(p,a)$ be a list of nodes
representing paths of the form $p\cdot a s'$ where
$s' \in \Supp(P(\last(p),a))$.
For each node (resp. node-action pair) we store %
  a value $\mctsvalue(p)$ (resp. $\mctsvalue(p,a)$) computed for node $p$ (resp. for playing $a$ from node $p$), meant to approximate $\Val_M^{H-|p|}(\last(p))$ (resp. $R(\last(p),a)+\sum_{s'}P(\last(p),a)(s')\Val^{H-|p|-1}_{M}(s')$),
  and a counter $\numsamples(p)$ (resp. $\numsamples(p,a)$), that keeps track
  of the number of iterations that selected node~$p$ (resp. that selected the action $a$ from $p$).
We add subscripts $i\geq 1$ to these notations to denote the number of previous iterations,
so that $\mctsvalue_i(p)$ is the value of $p$ obtained after $i$ iterations of MCTS,
among which $p$ was selected $\numsamples_i(p)$ times.
We also define $\total_i(p)$ and $\total_i(p,a)$ as shorthand for
respectively $\mctsvalue_i(p)\times\numsamples_i(p)$
and $\mctsvalue_i(p,a)\times\numsamples_i(p,a)$.
Each iteration consists of three phases. Let us describe these phases at iteration number $i$.

\paragraph{Selection phase} %
  Starting from the root node, MCTS descends through the existing search tree
  by choosing actions based on the current values and counters
  and by selecting next states stochastically according to the MDP.
  This continues until reaching a node $q$, either outside of the search tree
  or at depth $H$. %
  In the former case,
  the simulation phase is called to obtain a value $\mctsvalue_i(q)$
  that will be backpropagated along the path $q$. In the latter case,
  we use the exact value $\mctsvalue_i(q)=R_T(\last(q))$ instead.

  The action selection process
  needs to balance between the exploration of new paths
  and the exploitation of known, promising paths. %
  A popular way to balance both
  is the \emph{upper confidence bound for trees} (UCT)
  algorithm~\cite{DBLP:conf/ecml/KocsisS06},
  that
  interprets the action selection problem of each node of the MCTS tree
  as a bandit problem, and
  selects an action~$a^*$ in the set
    $\argmax_{a\in A}
    \left[
    \mctsvalue_{i-1}(p,a)
    +C
    \sqrt{\frac{\ln\left(\numsamples_{i-1}(p)\right)}
    {\numsamples_{i-1}(p, a)}}
    \right]$, for some constant $C$.

\paragraph{Simulation phase}
  In the simulation phase, the goal is to get an \emph{initial approximation} for the value of a node $p$, that will be refined in future iterations of MCTS.
  Classically, a sampling-based approach can be used, where one
  computes a fixed number $c\in\N$ of paths
  $p\cdot p'$ in $\FPaths_M^{H}(p)$.
  Then, one can compute $\mctsvalue_i(p) = \frac{1}{c}\sum_{p'}\Reward_M(p')$,
  and fix $\numsamples_i(p)$ to $1$.
  Usually, the samples are derived by selecting actions uniformly at random in the MDP.

  In our theoretical analysis of MCTS, we take a \emph{more general approach} to the simulation phase,
  defined by a finite domain $I\subseteq[0,1]$ and a function $f:\FPaths^{\leq H}_M\to\Dist(I)$
  that maps every path $p$ to a probability distribution on $I$.
  In this approach, the simulation phase simply draws
  a value $\mctsvalue_i(p)$ at random according to the distribution $f(p)$,
  and sets $\numsamples_i(p)=1$.

\paragraph{Backpropagation phase} From the value $\mctsvalue_{i}(p)$ obtained
  at a leaf node $p=s_0a_0s_1\dots s_h$ at depth $h$ in the search tree,
  let $\reward_i(p_{|k})=\sum_{l=k}^{h-1}R(s_l,a_l) + \mctsvalue_{i}(p)$ denote
  the reward associated with the path from node $p_{|k}$ to $p$ in the search tree.
  For $k$ from $0$ to $h-1$ we update the values according to
  $\mctsvalue_{i}(p_{|k})=\frac{\total_{i-1}(p_{|k}) + \reward_i(p_{|k})}{\numsamples_{i}(p_{|k})}\,.$
  The value $\mctsvalue_{i}(p_{|k},a_{k})$ is updated based on $\total_{i-1}(p_{|k},a_{k})$,
  $\reward_i(p_{|k})$ and $\numsamples_{i}(p_{|k},a_k)$ with the same formula.

\paragraph{Theoretical analysis} In the remainder of this section, we prove Theorem~\ref{thm:MCTS}, that provides
theoretical properties of
the MCTS algorithm with a general simulation phase (defined by some fixed $I$ and $f$).
This theorem was proven in \cite[Theorem 6]{DBLP:conf/ecml/KocsisS06} for a version
of the algorithm that called MCTS recursively until leaves were reached,
as opposed to the sampling-based approach that has become standard
in practice.
Note that sampling-based approaches are captured by our general description
of the simulation phase. Indeed, if the number of samples $c$ is set to $1$,
let $I$ be the set of rewards associated with paths of $\FPaths^{\leq H}_M$,
and let $f(p)$ be a probability distribution over $I$, such that for every reward $\Reward_M(p')\in I$,
$f(p)(\Reward_M(p'))$ is the probability of path $p'$ being selected with
a uniform action selections in $T(M,s_0,H)$, starting from the node $p$.
Then, the value $\mctsvalue_i(p)$ drawn
at random according to the distribution $f(p)$ corresponds to the reward of a random sample $p\cdot p'$
drawn in $\FPaths^{H}_M$.
If the number of samples $c$ is greater than $1$,
one simply needs to extend $I$ to be the set of average rewards
over $c$ paths, while $f(p)$ becomes a distribution over average rewards.

\begin{theorem}\label{thm:MCTS}%
  Consider an MDP $M$, a horizon $H$ and a state $s_0$.
  Let $V_n(s_0)$ be a random variable that represents the value $\mctsvalue_n(s_0)$ at the root
  of the search tree after $n$ iterations of the MCTS algorithm on $M$.
  Then, $|\E[V_n(s_0)]-\Val^H_{M}(s_0)|$ is bounded by $\mathcal O((\ln n)/n)$.
  Moreover, the failure probability
  $\Pr[\argmax_a\mctsvalue_n(s_0,a)\not\subseteq \opt^H_{M}(s_0)]$ converges to zero.
\end{theorem}

Following the proof scheme of \cite[Theorem 6]{DBLP:conf/ecml/KocsisS06},
this theorem is obtained from the results mentioned in Section~\ref{sec:UCB}.
To this end, every node $p$ of the search tree is considered to be an
instance of a bandit problem with non-stationary distributions.
Every time a node is selected, a step is processed in the corresponding
bandit problem.

Let $(\I_i(p))_{i\geq 1}$ be a sequence of iteration numbers for the MCTS algorithm
that describes when the node $p$ is selected, so that the simulation phase was used on $p$
at iteration number $\I_1(p)$, and so that the $i$-th selection
of node $p$ happened on the iteration number $\I_{i}(p)$.
We define sequences $(\I_i(p,a))_{i\in\N}$ similarly for node-action pairs.

For all paths $p$ and actions $a$, a payoff sequence $(x_{a,t})_{t\geq 1}$
of associated random variables $(X_{a,t})_{t\geq 1}$
is defined by $x_{a,t}=\reward_{\I_t(p,a)}(p)$.
Note that in the selection phase at iteration number $\I_t(p,a)$, $p$ must have been selected
and must be a prefix of length $k$ of the leaf node $p'$ reached in this iteration,
so that $\reward_{\I_t(p,a)}(p)$ is computed as $\reward_{\I_t(p,a)}(p'_{|k})$ in the backpropagation phase.
According to the notations of Section~\ref{sec:UCB}, for all $t\geq 1$ we have
$\numsamples_{\I_t(p)}(p)=t$, %
$\numsamples_{\I_t(p)}(p,a)=t_a$ and $\mctsvalue_{\I_t(p)}(p,a)=\overline x_{a,t_a}$.

Then, one can obtain Theorem~\ref{thm:MCTS}
by applying inductively the UCB1 results recalled in Appendix~\ref{app:UCB}
on the search tree in a bottom-up fashion.
Indeed, as the root $s_0$ is selected at every iteration, $\I_n(s_0)=n$
and $\mctsvalue_n(s_0)=\overline x_{n}$, while $\Val^H_{M}(s_0)$
corresponds to recursively selecting optimal actions by Proposition~\ref{prop:ValueIte}.

The main difficulty, and the difference our simulation phase brings compared with
the proof of \cite[Theorem 6]{DBLP:conf/ecml/KocsisS06}, lies in
showing that our payoff sequences $(x_{a,t})_{t\geq 1}$,
defined with an initial simulation step, still satisfy the drift conditions
of Definition~\ref{def:drift}.
We argue that this is true for all simulation phases defined by any $I$ and $f$:
\begin{lemma}\label{lm:MAB}
  For any MDP $M$, horizon $H$ and state $s_0$, %
  the sequences $(X_{a,t})_{t\geq 1}$
  satisfy the drift conditions.
\end{lemma}

Although the long-term guarantees of Theorem~\ref{thm:MCTS} hold
for any simulation phase independently of the MDP,
in practice one would expect better results from a good simulation, that gives
a value close to the real value of the current node.
Domain-specific knowledge can be used to obtain such simulations,
and also to guide the selection phase based on heuristics.
Our goal will be to preserve the theoretical guarantees of MCTS in the process.

\section{Symbolic advice for MCTS}

In this section, we introduce a notion of advice meant to guide the construction
of the Monte Carlo search tree. We argue that a symbolic approach
is needed in order to handle large MDPs in practice.
Let a \emph{symbolic advice} $\A$ be a logical formula
over finite paths whose truth value can be tested with an operator $\models$.

\begin{example}\label{ex:pacman-simulation-advice}
  A number of standard notions can fit this framework.
  For example, reachability and safety properties, LTL formul\ae\ over finite traces or
  regular expressions could be used.
  We will use a safety property for {\sc{Pac-Man}} as a example (see Figure~\ref{fig:grids}),
  by assuming that the losing states of the MDP should be avoided.
  This advice is thus satisfied by every path such that Pac-Man
  does not make contact with a ghost.
\end{example}

We denote by $\FPaths_M^H(\A)$ the set of paths $p\in\FPaths_M^H$ such that $p\models \A$. For a path $p\in \FPaths_M^{\leq H}$, we denote by $\FPaths_M^H(p,\A)$ the set of paths $p'\in\FPaths_M^H(p)$
such that $p'\models \A$.\footnote{In particular, for all $s\in S$, $\FPaths_M^H(s,\A)$ refers to the paths of length $H$ that start from $s$ and that satisfy $\A$.}

A \emph{nondeterministic strategy} is a function $\sigma : \FPaths_M \to 2^A$
that maps a finite path $p$ to a subset of $A$.
For a strategy $\sigma'$ and a nondeterministic strategy $\sigma$,
$\sigma'\subseteq \sigma$ if for all $p$,
$\Supp(\sigma'(p))\subseteq \sigma(p)$.
Similarly, a nondeterministic strategy for the environment is
a function $\tau : \FPaths_M\times A \to 2^S$
that maps a finite path $p$ and an action $a$ to a subset
of $\Supp(P(\last(p),a))$.
We extend the notations used for probabilistic strategies to nondeterministic
strategies in a natural way, so that $\FPaths_M^H(\sigma)$
and $\FPaths_M^H(\tau)$ denote the paths of length $H$ compatible with
the strategy $\sigma$ or $\tau$, respectively.

For a symbolic advice $\A$ and a horizon $H$, we define
a nondeterministic strategy $\sigma_\A^H$
and a nondeterministic strategy $\tau_\A^H$ for the environment such that
for all paths $p$ with $|p|< H$,
$$\sigma_\A^H(p)=\{a\in A\mid
\exists s\in S, \exists p'\in\FPaths_M^{H-|p|-1}(s), p\cdot as \cdot p'\models \A\}\,,$$
$$\tau_\A^H(p,a)=\{s\in S\mid
\exists p'\in\FPaths_M^{H-|p|-1}(s), p\cdot as \cdot p'\models \A\}\,.$$
The strategies $\sigma_\A^H$ and $\tau_\A^H$ can be defined arbitrarily
on paths $p$ of length at least $H$, for example with $\sigma_\A^H(p)=A$ and $\tau_\A^H(p,a)=\Supp(P(\last(p),a))$
for all actions $a$.
Note that by definition, $\FPaths_M^H(s,\A)=\FPaths_M^H(s,\sigma_\A^H)
\cap\FPaths_M^H(s,\tau_\A^H)$ for all states $s$.

Let $\top$ (resp. $\bot$) denote the universal advice (resp. the empty advice)
satisfied by every finite path (resp. never satisfied), and let $\sigma_\top$
and $\tau_\top$ (resp. $\sigma_\bot$ and $\tau_\bot$) be the associated
nondeterministic strategies.
We define a class of advice that can be enforced against an adversarial environment
by following a nondeterministic strategy,
and that are minimal in the sense that paths that are not compatible with this
strategy are not allowed.

\begin{definition}[Strongly enforceable advice]\label{def:Advice}
A symbolic advice $\A$ is called a strongly enforceable advice from a state $s_0$ and for a horizon $H$ if
there exists a nondeterministic strategy $\sigma$ such that
$\FPaths_M^H(s_0,\sigma)=\FPaths_M^H(s_0,\A)$,
and such that $\sigma(p)\neq\emptyset$
for all paths $p\in\FPaths_M^{\leq H-1}(s_0,\sigma)$.
\end{definition}

Note that Definition~\ref{def:Advice} ensures that paths that follow $\sigma$
can always be extended into longer paths that follow $\sigma$.
This is a reasonable assumption to make for a nondeterministic strategy
meant to enforce a property.
In particular, $s_0$ is a path of length 0 in $\FPaths_M^{0}(s_0,\sigma)$,
so that $\sigma(s_0)\neq\emptyset$ and so that by induction $\FPaths_M^i(s_0,\sigma)\neq\emptyset$
for all $i\in[0,H]$.

\begin{lemma}\label{lm:enforceAdv}
  Let $\A$ be a strongly enforceable advice from $s_0$ with horizon $H$.
  It holds that $\FPaths_M^H(s_0,\sigma_\A^H)=\FPaths_M^H(s_0,\A)$.
  Moreover, for all paths $p\in\FPaths_M^{\leq H-1}(s_0)$ and all actions $a$,
  either $\tau_\A^H(p,a)=\tau_\top(p,a)$ or $\tau_\A^H(p,a)=\tau_\bot(p,a)$.
  Finally, for all paths $p$ in $\FPaths_M^{\leq H-1}(s_0,\sigma_\A^H)$,
  $\sigma_\A^H(p)\neq\emptyset$ and $a\in\sigma_\A^H(p)$ if and only if $\tau_\A^H(p,a)=\tau_\top(p,a)$.
\end{lemma}
\begin{proof}%
  We have $\FPaths_M^H(s_0,\A)=\FPaths_M^H(s_0,\sigma_\A^H) \cap\FPaths_M^H(s_0,\tau_\A^H)$
  for any advice $\A$.
  Let us prove that $\FPaths_M^H(s_0,\sigma_\A^H)\subseteq\FPaths_M^H(s_0,\A)$ for
  a strongly enforceable advice $\A$ of associated strategy $\sigma$.
  Let $p=p'\cdot as$ be a path in $\FPaths_M^H(s_0,\sigma_\A^H)$.
  By definition of $\sigma_\A^H$,
  there exists $s'\in S$ such that $p'\cdot a s'\models\A$, so that
  $p'\cdot a s'\in \FPaths_M^H(s_0,\A)=\FPaths_M^H(s_0,\sigma)$. Since $s\in\Supp(P(\last(p'),a))$,
  $p=p'\cdot as$ must also belong to
  $\FPaths_M^H(s_0,\sigma)=\FPaths_M^H(s_0,\A)$.

  Consider a path $p$ and an action $a$ such that $|p|< H$.
  We want to prove that either all stochastic transitions starting from $(p,a)$
  are allowed by $\A$, or none of them are.
  By contradiction, let us assume that there exists $s_1$ and $s_2$ in $\Supp(P(\last(p),a))$ such that
  for all $p'_1\in\FPaths_M^{H-|p|-1}(s_1)$, $p\cdot as_1 \cdot p'_1\not\models \A$,
  and such that there exists $p'_2\in\FPaths_M^{H-|p|-1}(s_2)$
  with $p\cdot as_2 \cdot p'_2\models \A$.
  From $p\cdot as_2 \cdot p'_2\models \A$,
  we obtain $p\cdot as_2 \cdot p'_2\in \FPaths_M^H(\sigma)$, so that
  $p\cdot as_2$ is a path that follows $\sigma$.
  Then, $p\cdot as_1$ is a path that follows $\sigma$ as well.
  It follows that $\sigma(p\cdot as_1)\neq\emptyset$, and $p\cdot as_1$
  can be extended in to a path $p\cdot as_1 p'_3\in\FPaths_M^H(\sigma)$.
  This implies the contradiction $p\cdot as_1 p'_3\models\A$.

  Finally, consider a path $p$ in $\FPaths_M^{\leq H-1}(s_0,\sigma_\A^H)$.
  By the definitions of $\sigma_\A^H$ and $\tau_\A^H$, $a\in\sigma_\A^H(p)$ if and only if $\tau_\A^H(p,a)\neq\emptyset$, so that $\tau_\A^H(p,a)=\tau_\top(p,a)$.
  Then, let us write $p=p'\cdot a s$. From $p\in\FPaths_M^{\leq H-1}(s_0,\sigma_\A^H)$ we get
  $a\in\sigma_\A^H(p')$, so that $s\in\tau_\A^H(p',a)$, and therefore $\sigma_\A^H(p)\neq\emptyset$.
\end{proof}

A strongly enforceable advice is encoding a notion of guarantee, as
$\sigma_\A^H$ is a winning strategy for the
reachability objective on $T(M,s_0,H)$ defined by the set $\FPaths^H_M(\A)$.

We say that the strongly enforceable advice $\A'$ is \emph{extracted} from a symbolic advice $\A$
for a horizon $H$ and a state $s_0$
if $\A'$ is the greatest part of $\A$ that can be guaranteed for the horizon $H$ starting from $s_0$,
\textit{i.e.}~if $\FPaths^H_M(s_0,\A')$ is the greatest subset of $\FPaths^H_M(s_0,\A)$
such that $\sigma_{\A'}^H$ is a winning strategy for the reachability
objective $\FPaths^H_M(s_0,\A)$ on $T(M,s_0,H)$.
This greatest subset always exists
because if $\A'_1$ and
$\A'_2$ are strongly enforceable advice in $\A$, then $\A'_1\cup\A'_2$
is strongly enforceable by union of the nondeterministic strategies associated
with $\A'_1$ and $\A'_2$.
However, this greatest subset may be empty, and as $\bot$ is not a strongly enforceable advice we say that
in this case $\A$ cannot be enforced from $s_0$ with horizon $H$.

\begin{example}\label{ex:pacman-selection-advice}
  Consider a symbolic advice $\A$ described by the safety property for {\sc{Pac-Man}} of Example~\ref{ex:pacman-simulation-advice}.
  For a fixed horizon $H$, the associated nondeterministic strategies
  $\sigma_\A^H$ and $\tau_\A^H$ describe action choices and stochastic transitions
  compatible with this property.
  Notably, $\A$ may not be a strongly enforceable advice,
  as there may be situations $(p,a)$ where some
  stochastic transitions lead to bad states and some do not.
  In the small grid of Figure~\ref{fig:grids}, the path of length 1 that corresponds to Pac-Man
  going left and the red ghost going up is allowed by the advice $\A$, but not by any
  safe strategy for Pac-Man as there is a possibility of losing by playing left.
  If a strongly enforceable advice $\A'$ can be extracted from $\A$,
  it is a more restrictive safety property,
  where the set of bad states is obtained as the attractor~\cite[Section~2.3]{DBLP:books/cu/11/Loding11}
  for the environment towards the bad states defined in $\A$.
  In this setting, $\A'$ corresponds to playing according to a strategy for Pac-Man
  that ensures not being eaten by adversarial ghosts for the next $H$ steps.
\end{example}

\begin{definition}[Pruned MDP]
  For an MDP $M=(S,A,P,R,R_T)$ a horizon $H\in\N$, a state $s_0$
  and an advice $\A$, let the pruned unfolding $T(M,s_0,H,\A)$
  be defined as a sub-MDP of $T(M,s_0,H)$ that contains exactly all paths
  in $\FPaths_M^H(s_0)$ satisfying $\A$.
  It can be obtained by removing all action transitions
  that are not compatible with $\sigma_\A^H$, and all stochastic transitions
  that are not compatible with $\tau_\A^H$.
  The distributions $P(p,a)$ are then normalised over the
  stochastic transitions that are left.
\end{definition}

Note that by Lemma~\ref{lm:enforceAdv}, if $\A$ is a strongly enforceable advice then
$\tau_\A^H(p,a)=\tau_\top(p,a)$ for all paths $p$ in $\FPaths_M^{\leq H-1}(s_0,\sigma_\A^H)$,
so that
the normalisation step for the distributions $P(p,a)$ is not needed.
It follows that for all nodes $p$ in $T(M,s_0,H,\A)$ and all actions $a$,
the distributions $P(p,a)$ in $T(M,s_0,H,\A)$ are the same as in $T(M,s_0,H)$.
Thus, for all strategies $\sigma$ in $T(M,s_0,H,\A)$, $\Val^{H}_{T(M,s_0,H,\A)}(s_0,\sigma)=\Val^{H}_{T(M,s_0,H)}(s_0,\sigma)$,
so that $\Val^{H}_{T(M,s_0,H,\A)}(s_0)\leq \Val^{H}_{T(M,s_0,H)}(s_0)=\Val^H_{M}(s_0)$ by Lemma~\ref{lm:unfolding}.

\begin{definition}[Optimality assumption]\label{lm:opt-ass}
  An advice $\A$ satisfies the optimality assumption for horizon $H$ if
  $\sigma^{H,*}_{M,s}\subseteq\sigma_\A^H$ for all $s\in S$,
  where $\sigma^{H,*}_{M,s}$ is an
  optimal strategy for the expected total reward of horizon $H$
  at state $s$.
\end{definition}

\begin{lemma}\label{lm:opt-adv}
  Let $\A$ be a strongly enforceable advice that satisfies the optimality assumption.
  Then, $\Val^{H}_M(s_0)$ equals $\Val^{H}_{T(M,s_0,H,\A)}(s_0)$.
  Moreover, $\opt^{H}_{T(M,s_0,H,\A)}(s_0)\subseteq\opt^{H}_M(s_0)$.
\end{lemma}
\begin{proof}
By the optimality assumption $\sigma^{H,*}_{M,s}$ is a strategy that can be followed in $T(M,s_0,H,\A)$.
Indeed, from a path $p$ in $\FPaths_M^{\leq H-1}(s_0,\sigma_\A^H)$ any action $a$ in the support of
$\sigma^{H,*}_{M,s}(\last(p))$ satisfies $a\in\sigma_\A^H(p)$.
Thus, by Lemma~\ref{lm:unfolding} $\Val^H_{M}(s_0)=\Val^{H}_{T(M,s_0,H)}(s_0,\sigma^{H,*}_{M,s})=\Val^{H}_{T(M,s_0,H,\A)}(s_0,\sigma^{H,*}_{M,s})$.
By definition of the optimal expected total reward, $\Val^{H}_{T(M,s_0,H,\A)}(s_0,\sigma^{H,*}_{M,s})\leq \Val^{H}_{T(M,s_0,H,\A)}(s_0)$, so that $\Val^{H}_M(s_0)=\Val^{H}_{T(M,s_0,H,\A)}(s_0)$.
Let $a$ be an action in $\opt^{H}_{T(M,s_0,H,\A)}(s_0)$. There exists an optimal strategy
$\sigma$ that maximises $\Val^{H}_{T(M,s_0,H,\A)}(s_0,\sigma)$ so that
$\sigma(s_0)=a$.
It follows from $\Val^{H}_{T(M,s_0,H)}(s_0)=\Val^{H}_{T(M,s_0,H,\A)}(s_0)$ that
$\sigma$ is also an optimal strategy in $T(M,s_0,H)$, so that $a\in\opt^{H}_{T(M,s_0,H)}(s_0)=\opt^{H}_M(s_0)$
by Lemma~\ref{lm:unfolding}.
\end{proof}

\begin{example}
  Let $\A'$ be a strongly enforceable safety advice for {\sc{Pac-Man}} as described in Example~\ref{ex:pacman-selection-advice}.
  Assume that visiting a bad state leads to an irrecoverably bad reward,
  so that taking an unsafe action (\textit{i.e.}~an action such that there is a non-zero probability of losing
  associated with all Pac-Man strategies) is always worse (on expectation) than taking a safe action.
  Then, the optimality assumption holds for the advice $\A'$.
  This can be achieved by giving a penalty score for losing that is low enough.
\end{example}

\subsection{MCTS under symbolic advice}\label{sec:MCTSAdvice}

We will augment the MCTS algorithm using two advice: a selection advice $\varphi$
to guide the MCTS tree construction, and a simulation advice $\psi$
to prune the sampling domain.
We assume that the selection advice is a strongly enforceable advice that satisfies
the optimality assumption. %
Notably, we make no such assumption for the simulation advice,
so that any symbolic advice can be used.

\paragraph{Selection phase under advice} We use the advice $\varphi$
  to prune the tree according to $\sigma_\varphi^H$.
  Therefore, from any node $p$ our version of UCT selects an action~$a^*$ in the set
    $$\argmax_{a\in \sigma_\varphi^H(p)}
    \left[
    \mctsvalue(p,a)
    +C
    \sqrt{\frac{\ln\left(\numsamples(p)\right)}
    {\numsamples(p, a)}}
    \right]\,.$$
\paragraph{Simulation phase under advice} For the simulation phase,
  we use a sampling-based
  approach biased by the simulation advice: paths are sampled by
  picking actions uniformly at random in the
  pruned MDP $T(M,s_0,H,\psi)$, with a fixed prefix $p$ defined
  by the current node in the search tree.
  This can be interpreted as a probability distribution
  over $\FPaths^H_M(p,\psi)$.
  If $p\notin T(M,s_0,H,\psi)$, the simulation phase outputs a value of $0$
  as it is not possible to satisfy $\psi$ from $p$.
  Another approach that does not require computing
  the pruned MDP repeats the following steps for a bounded number of time
  before returning $0$ if no valid sample is found:
  \begin{enumerate}
    \item Pick a path $p\cdot p'\in \FPaths^H_M(p)$ using a uniform sampling method;
    \item If $p\cdot p'\not\models\psi$, reject and try again, otherwise output $p'$ as a sample.
  \end{enumerate}
  We compute $\mctsvalue_i(p)$ by averaging the rewards of these samples.

\paragraph{Theoretical analysis} We show that the theoretical guarantees of the MCTS algorithm
developed in Section~\ref{sec:MCTSSimulation} are maintained by the MCTS algorithm
under symbolic advice.

\begin{theorem}\label{thm:advMCTS}
  Consider an MDP $M$, a horizon $H$ and a state $s_0$.
  Let $V_n(s_0)$ be a random variable that represents the value $\mctsvalue_n(s_0)$
  at the root
  of the search tree after $n$ iterations of the MCTS algorithm under
  a strongly enforceable advice $\varphi$ satisfying
  the optimality assumption and a simulation advice $\psi$.
  Then, $|\E[V_n(s_0)]-\Val^H_{M}(s_0)|=\mathcal O((\ln n)/n)$.
  Moreover, the failure probability
  $\Pr[\argmax_a\mctsvalue_n(s_0,a)\not\subseteq \opt^H_{M}(s_0)]$ converges to zero.
\end{theorem}

In order to prove Theorem~\ref{thm:advMCTS},
we argue that running MCTS under a selection advice $\varphi$ and
a simulation advice $\psi$ is equivalent to running the MCTS algorithm
of Section~\ref{sec:MCTSSimulation} on the pruned MDP $T(M,s_0,H,\varphi)$, with
a simulation phase defined using the advice $\psi$.

The simulation phase biased by $\psi$ can be described
in the formalism of Section~\ref{sec:MCTSSimulation},
with a domain $I=\{\frac{1}{c}\sum_{i=1}^c\Reward_M(p_i)\mid p_1,\dots,p_c\in \FPaths_{T(M,s_0,H,\varphi)}^{\leq H}\}$,
and a mapping $f_{\psi}$
from paths $p$ in $\FPaths_{T(M,s_0,H,\varphi)}^{\leq H}$ to
a probability distribution on $I$ describing the outcome
of a sampling phase launched
from the node $p$.
Formally, the weight of $\frac{1}{c}\sum_{i=1}^c\Reward_M(p_i)\in I$ in $f(p)$ is
the probability of sampling the sequence of paths $p_1,\dots,p_c$ in
the simulation phase under advice launched from $p$.
Then, from Theorem~\ref{thm:MCTS} we obtain convergence properties
of MCTS under symbolic advice
towards the value and optimal strategy in the pruned MDP,
and Lemma~\ref{lm:opt-adv} lets us conclude the proof of Theorem~\ref{thm:advMCTS}
as those values and strategies are maintained in $M$ by the optimality assumption.
In particular, the failure probability $\Pr[\argmax_a\mctsvalue_n(s_0,a)\not\subseteq \opt^H_{M}(s_0)]$
is upper bounded by $\Pr[\argmax_a\mctsvalue_n(s_0,a)\not\subseteq \opt^H_{T(M,s_0,H,\varphi)}(s_0)]$
since $\opt^H_{T(M,s_0,H,\varphi)}\subseteq\opt^H_{M}(s_0)$.

\subsection{Using satisfiability solvers}

We will now discuss the use of general-purpose solvers to implement
symbolic advice according to the needs of MCTS.

A symbolic advice $\A$ describes a finite set of paths
in $\FPaths^H_M$, and as such can be encoded as
a Boolean formula over a set of variables $V$,
such that satisfying assignments
$v:V\to\{\mathsf{true},\mathsf{false}\}$ are in bijection with
paths in $\FPaths^H_M(\A)$.

If a symbolic advice is described in Linear Temporal Logic,
and a symbolic model of the MDP $M$ is available,
one can encode $\A$ as a Boolean formula of size linear in the size of the LTL formula and $H$~\cite{lmcs:2236}.

\begin{example}
  In practice, one can use Boolean variables to encode the positions of Pac-Man and ghosts in the next $H$ steps
  of the game, then construct a CNF formula with clauses that encode the game rules and clauses that enforce the advice.
  The former clauses are implications such as ``if Pac-Man is in position $(x,y)$ and plays the action $a$, then it must be in position $(x',y')$ at the next game step'', while the latter clauses state that the position of Pac-Man should never be equal to the position of one of the Ghosts.
\end{example}

\paragraph{On-the-fly computation of a strongly enforceable advice}
A direct encoding of a strongly enforceable advice may prove
impractically large.
We argue for an on-the-fly computation of $\sigma_\A^H$ instead,
in the particular case where the strongly enforceable advice
is extracted from a symbolic advice $\A$
with respect to the initial state $s_0$ and with horizon $H$.
\begin{lemma}\label{lm:qbf-advice}
  Let $\A'$ be a strongly enforceable advice extracted from $\A$ for horizon $H$. Consider a node $p$
  at depth $i$ in $T(M,s_0,H,\A')$, for all $a_0\in A$,
  $a_0\in\sigma_{\A'}^H(p)$ if and only if
  $$\forall s_1 \exists a_1 \forall s_2 \dots \forall s_{H-i+1},
  \ p\cdot a_0s_1a_1s_2\dots s_{H-i+1} \models \A\,,$$
  where actions are quantified over $A$ and every $s_k$ is
  quantified over $\Supp(P(s_{k-1},a_{k-1}))$.
\end{lemma}
\begin{proof}%
  The proof is a reverse induction on the depth $i$ of $p$.
  For the initialisation step with $i=H$, let us prove that
  $\forall s_1,\ p\cdot a_0s_1\models\A$ if and only if $a_0\in\sigma_{\A'}^H(p)$.
  On the one hand, if $\A$ is guaranteed by playing $a_0$ from $p$, then $a_0$ must be
  allowed by the greatest strongly enforceable subset of $\A$.
  On the other hand, $a_0\in\sigma_{\A'}^H(p)$ implies $\forall s_1,\ p\cdot a_0s_1\models\A'$ as $\A'$ is strongly enforceable, and finally $\A'\Rightarrow\A$.
  We now assume the property holds for $1\leq i\leq H$, and prove it for $i-1$.
  If $a_0\in\sigma_{\A'}^H(p)$, then for all $s_1$ we have $s_1\in\tau_{\A'}^H(p,a_0)$,
  so that there exists $a_1$ with $a_1\in\sigma_{\A'}^H(p\cdot a_0s_1)$.
  As $p\cdot a_0s_1$ is at depth $i$ we can conclude that
  $\forall s_1 \exists a_1 \forall s_2 \dots \forall s_{H-i+1},
  \ p\cdot a_0s_1a_1s_2\dots s_{H-i+1} \models \A$ by assumption.
  For the converse direction, the alternation of quantifiers states that $\A$
  can be guaranteed from $p$ by some deterministic strategy that starts
  by playing $a_0$, and therefore $a_0$ must be allowed by the strongly enforceable
  advice extracted from $\A$.
\end{proof}

Therefore, given a Boolean formula encoding $\A$, one can use
a Quantified Boolean Formula (QBF) solver to compute~$\sigma_{\A'}^H$, the strongly enforceable advice extracted from $\A$:
this computation can be used whenever MCTS performs an action selection step under the advice $\A'$,
as described in Section~\ref{sec:MCTSAdvice}.

The performance of this approach will crucially depend on the number of alternating quantifiers,
and in practice one may limit themselves to a smaller depth $h<H-i$ in this step,
so that safety is only guaranteed for the next $h$ steps.

Some properties can be inductively guaranteed, so that satisfying the QBF formula
of Lemma~\ref{lm:qbf-advice} with
a depth $H-i=1$ is enough to guarantee the property globally. For example, if there always exists
an action leading to states that are not bad,
it is enough to check for safety locally with a depth of 1.
This is the case in {\sc{Pac-Man}} for a deadlock-free layout when there is only one ghost.

\paragraph{Weighted sampling under a symbolic advice}
Given a symbolic advice $\A$ as a Boolean formula, and a probability distribution $w\in \Dist(\FPaths_M^H)$,
our goal is to sample paths of $M$ that satisfy $\A$ with respect to $w$.\footnote{The probability of a path $p$
being sampled should be equal to $w(p)/\sum_{p'|p'\models\A}w(p')$.}
Let $\omega$ denote a weight function over Boolean assignments that matches $w$.
This reduces our problem to the weighted sampling of satisfying assignments in a Boolean formula.
An exact solver for this problem may not be efficient, but one can
use the techniques of \cite{DBLP:conf/aaai/ChakrabortyFMSV14} to perform approximate sampling in polynomial time:
\begin{proposition}[{\cite{DBLP:conf/aaai/ChakrabortyFMSV14}}]\label{prop:WeightGen}
Given a CNF formula $\A$, a tolerance $\epsilon>0$ and a weight function $\omega$, we can construct a probabilistic algorithm which outputs a satisfying assignment $z$ such that for all $y$ that satisfies $\A$:
$$\frac{\omega(y)}{(1+\epsilon)\sum_{x\models \psi}\omega(x)}\leq Pr[z=y] \leq \frac{(1+\epsilon)\omega(y)}{\sum_{x\models \psi}\omega(x)}$$

The above algorithm occasionally `fails' (outputs no assignment even though there are satisfying assignments) but its failure probability can be bounded by any given $\delta$. Given an oracle for $SAT$, the above algorithm runs in time polynomial in $\ln\left(\frac{1}{\delta}\right)$, $|\psi|$, $\frac{1}{\epsilon}$ and $r$ where $r$ is the ratio between highest and lowest weight according to $\omega$.
\end{proposition}
In particular, this algorithm uses $\omega$ as a black-box, and thus
does not require precomputing the probabilities of all paths satisfying $\A$.
In our particular application of Proposition~\ref{prop:WeightGen}, the value $r$
can be bounded by $\left(\frac{p_{\max}|A|}{p_{\min}}\right)^H$ where
$p_{\min}$ and $p_{\max}$ are the smallest and greatest probabilities for
stochastic transitions in $M$.

Note that if we wish to sample from a given node $p$ of the search tree, we can force $p$ as a mandatory
prefix of satisfying assignments by fixing the truth value of relevant variables in the Boolean formula.

\section{A {\sc{Pac-Man}} case study}

We performed our experiments on the multi-agent game {\sc{Pac-Man}}, using the code of \cite{BerkeleyCS188}. The ghosts can have different strategies where they take actions based on their own position as well as position of Pac-Man. In our experiments, we used two different types of ghosts, %
the \emph{random ghosts} (in green) always choose an action uniformly at random from the legal actions available,
while the \emph{directional ghosts} (in red) take the legal action that minimises the Manhattan distance to Pac-Man with probability 0.9, and move randomly otherwise.

The game can be seen as a Markov decision process, where states encode a position for each agent\footnote{The last action played by ghosts should be stored as well, as they are not able to reverse their direction.} and for
the food pills in the grid, where actions encode individual Pac-Man moves, and where stochastic transitions
encode the moves of ghosts according to their probabilistic models.
For each state and action pair, we define a reward based on the score gained or lost by this move, as explained in the caption of Figure~\ref{fig:grids}.
We also assign a terminal reward to each state, so as to allow MCTS to compare paths of length $H$
which would otherwise obtain the same score.
Intuitively, better terminal rewards are given to states where Pac-Man
is closer to the food pills and further away from the ghosts, so that terminal rewards play the role
of a static evaluation of positions.

\paragraph{Experiments} We used a receding horizon $H=10$.
The baseline is given by a standard implementation of the algorithm described
in Section~\ref{sec:MCTSSimulation}. %
A search tree is constructed with a
maximum depth $H$, for $100$ iterations, so that the search tree constructed by the MCTS algorithm
contains up to $100$ nodes. At the first selection of every node,
$100$ samples are obtained by using a uniform policy. %
Overall, this represents a tiny portion of the tree unfolding of depth $10$,
which underlines the importance of properly guiding the search to the most interesting neighborhoods.
As a point of reference, we also had human players take control of Pac-Man,
and computed the same statistics.
The players had the ability to slow down the game as they saw fit,
as we aimed for a comparison between the quality of the strategical decisions made by these approaches,
and not of their reaction speeds.

We compare these baselines with the algorithm of Section~\ref{sec:MCTSAdvice},
using the following advice.
The \emph{simulation advice} $\psi$ that we consider is defined as a safety property
satisfied by every path such that Pac-Man
does not make contact with a ghost, as in Example~\ref{ex:pacman-simulation-advice}.
We provide a Boolean formula encoding $\psi$, so that one can use a SAT solver
to obtain samples, or sampling tools as described in Proposition~\ref{prop:WeightGen},
such as {\sc{WeightGen}}~\cite{DBLP:conf/aaai/ChakrabortyFMSV14}.
We use {\sc{UniGen}}~\cite{DBLP:conf/tacas/ChakrabortyFMSV15}
to sample almost uniformly over the satisfying assignments of $\psi$.\footnote{The distribution over path
is slightly different than when sampling uniformly over actions in the pruned MDP $T(M,s_0,H,\psi)$, but {\sc{UniGen}} enjoys better performances than {\sc{WeightGen}}.}

From this simulation advice, we extract whenever possible a strongly enforceable \emph{selection advice} $\varphi$
that guarantees that Pac-Man will not make contact with a ghost,
as described in Example~\ref{ex:pacman-selection-advice}.
If safety cannot be enforced, $\top$ is used as a selection advice, so that no pruning is performed.
This is implemented by using the Boolean formula $\psi$ in a QBF solver
according to Lemma~\ref{lm:qbf-advice}.
For performance reasons, we guarantee safety for a smaller horizon $h<10$,
that we fixed at $3$ in our experiments.

Several techniques were used to reduce the state-space of the MDP
in order to obtain
smaller formul\ae. For example,
a ghost that is too far away with respect to $H$ or $h$ can be safely ignored,
and the current positions of the food pills is not relevant for safety.

\paragraph{Results}
For each experiment, we ran $100$ games in a high-end cluster using AMD Opteron Processors 6272 at 2.1 GHz.
\begin{table}
  \centering
  \resizebox{\linewidth}{!}{
  \begin{tabular}{|c|c|c|c|c|c|c|c|}
    \hline
    Grid & Ghosts & Algorithm & win & loss & draw & food & score \\
    \hline
    & & MCTS & 17 & 59 & 24 & 16.65 & -215.32 \\
    \cline{3-8}
     & & MCTS+Selection advice & 25 & 54 & 21 & 17.84 & -146.44 \\
    \cline{3-8}
     & 4 x Random & MCTS+Simulation advice & 71 & 29 & 0 & 22.11 & 291.80 \\
    \cline{3-8}
     &  & MCTS+both advice & 85 & 15 & 0 & 23.42 & 468.74 \\
    \cline{3-8}
     9 x 21 & & Human & 44 & 56 & 0 & 18.87 & 57.76 \\
    \cline{2-8}
     &  & MCTS & 11 & 85 & 4 & 14.86 & -339.99 \\
    \cline{3-8}
     & 1 x Directional & MCTS+Selection advice & 16 & 82 & 2 & 15.25 & -290.6 \\
    \cline{3-8}
     & + & MCTS+Simulation advice & 27 & 70 & 3 & 17.14 & -146.79 \\
    \cline{3-8}
     & 3 x Random & MCTS+both advice & 33 & 66 & 1 & 17.84 & -92.47 \\
    \cline{3-8}
     &  & Human & 24 & 76 & 0 & 15.10 & -166.28 \\
        \hline
    27 x 28 & 4 x Random & MCTS & 1 & 10 & 89 & 14.85 & -182.77 \\
    \cline{3-8}
      & & MCTS+both advice & 95 & 5 & 0 & 24.10 & 517.04 \\
    \hline
  \end{tabular}}
  \caption{Summary of experiments with different ghost models, algorithms and grid size. The win, loss and draw columns denote win/loss/draw rates in percents
  (the game ends in a draw after 300 game steps).
  The food eaten column refers to the number of food pills eaten on average, out of 25 food pills in total. Score refers to the average score obtained over all runs. %
  }
  \label{table:exp}
\end{table}
A summary of our results is displayed in Table~\ref{table:exp}. We mainly use the number of games won out of 100
to evaluate the performance of our algorithms.\footnote{We do not evaluate the accuracy in terms of making optimal choices because those cannot be computed due to the size of the MDPs (about $10^{16}$ states).}
In the small grid with four random ghosts, the baseline MCTS algorithm wins 17\% of games.
Adding the selection advice results in a slight increase of the win rate to 25\%.
The average score is improved as expected, but even if one ignores the $\pm 500$ score associated with a win or a loss,
we observe that more food pills were eaten on average as well.
The simulation advice provides in turn a sizeable increase in both win rate (achieving 71\%)
and average score.
Using both advice at the same time gave the best results overall, with a win rate of 85\%.
The same observations can be made in other settings as well,
either with a directional ghost model or on a large grid.
Moreover, the simulation advice significantly reduces the number of game turns Pac-Man needs to win,
resulting in fewer game draws, most notably on the large grid.

In the baseline experiments without any advice, computing the next action played by Pac-Man from a given state takes 200 seconds of runtime on average.\footnote{This holds for both the small and large grids, as in both cases we consider the next 10 game steps only, resulting in MCTS trees of similar size.} Comparatively, the algorithm with both advice is about three times slower than the baseline.
If we make use of the same time budget in the standard MCTS algorithm (roughly increasing the number of
nodes in the MCTS tree threefold), the win rate climbs to 26\%, which is still significantly below the 85\% win rate achieved with advice.
Moreover, while this experiment was not designed to optimise the performance of these approaches
in terms of computing time, we were able to achieve a performance of 5s per move on a standard laptop
by reducing the number of iterations and samples in MCTS.
This came at the cost of a decreased win-rate of 76\% with both advice.
Further code improvements \textit{e.g.}~using parallelism as in \cite{Chaslot08} could reasonably lead to real-time performances.

Supplementary material is available at \href{http://di.ulb.ac.be/verif/debraj/pacman/}{http://di.ulb.ac.be/verif/debraj/pacman/}.

\section{Conclusion and future works}

In this paper, we have introduced the notion of symbolic advice to guide the selection and the simulation phases of the MCTS algorithm. We have identified sufficient conditions to preserve the convergence guarantees offered by the MCTS algorithm while using symbolic advice.  We have also explained how to implement them using SAT and QBF solvers in order to apply symbolic advice to large MDP defined symbolically rather than explicitly. We believe that the generality, flexibility and precision offered by logical formalism to express symbolic advice in MCTS can be used as the basis of a methodology to systematically inject domain knowledge into MCTS. We have shown that domain knowledge expressed as simple symbolic advice (safety properties) improves greatly the efficiency of the MCTS algorithm in the {\sc Pac-Man} application. This application is challenging as the underlying MDPs have huge state spaces, \textit{i.e.}~up to $10^{23}$ states. In this application, symbolic advice allow the MCTS algorithm to reach or even surpass human level in playing.

As further work, we plan to offer a compiler from LTL to symbolic advice, in order to automate their integration in the MCTS algorithm for diverse application domains. We also plan to work on the efficiency of the implementation. So far, we have developed a prototype implementation written in Python (an interpreted language). This implementation cannot be used to evaluate performances in absolute terms but it was useful to show that if the same amount of resources is allocated to the two algorithms the one with advice performs much better. We believe that by
using a well-optimised code base
and by exploiting parallelism, we should be able to apply our algorithm in real-time and preserve the level of quality reported in the experimental section. Finally, we plan to study how learning can be incorporated in our framework. One natural option is to replace the static reward function used after $H$ steps by a function learned from previous runs of the algorithm and implemented using a neural network (as it is done in AlphaGo~\cite{DBLP:journals/nature/SilverHMGSDSAPL16} for example).

We thank Gilles Geeraerts for fruitful discussions in the early phases of this work.

\bibliographystyle{plainurl}%

\newpage
\appendix

\section{Markov decision processes}\label{app:mdp}

One can make an MDP strongly aperiodic without changing the optimal expected average reward
and its optimal strategies with the following transition:
\begin{definition}[Aperiodic transformation]\cite[Section 8.5.4]{DBLP:books/wi/Puterman94}
  For an MDP $M=(S,A,P,R)$, we define a new MDP $M_\alpha=(S,A,P_\alpha,R_\alpha)$ for $0<\alpha<1$, with
  $R_\alpha(s,a)=R(s,a)$, $P_\alpha(s,a)(s)=\alpha+(1-\alpha)P(s,a)(s)$ and $P_\alpha(s,a)(s')=(1-\alpha) P(s,a)(s')$.
  Notice that $M_\alpha$ is strongly aperiodic.
\end{definition}

Every finite path in $M$ is also in $M_{\alpha}$. Thus for a strategy $\widehat \sigma$ in $M_\alpha$, there is a $\sigma$ in $M$ whose domain is restricted to the paths in $M$.

\begin{proposition}\cite[Section 8.5.4]{DBLP:books/wi/Puterman94}
  Let $M$ be an MDP. $M_\alpha$ is a new MDP generated by applying the aperiodic transformation mentioned above. Then the set of memoryless strategies that optimises the expected average reward in $M_{\alpha}$ is the same as the set of memoryless strategies the optimises the expected average reward in $M$.
  Also from any $s$, $\Val_M(s)=\Val_{M_\alpha}(s)$.
\end{proposition}

\begin{definition}[Finite horizon unfolding of an MDP]
  For an MDP $M=(S,A,P,R,R_T)$, a horizon depth $H\in\mathbb N$ and a state $s_0$, the unfolding of $M$
  from $s_0$ and with horizon $H$ is a tree-shaped MDP defined as
  $T(M,s_0,H)=(S'=S_0\cup\dots\cup S_H,A,P',R',R'_T)$, where
  for all $i\in[0,H]$, $S_i=\FPaths^i(s_0)$.
  The mappings $P'$, $R'$ and $R'_T$ are inherited from $P$, $R$ and $R'_T$ in a natural way
  with additional self-loops at the leaves of the unfolding,
  so that for all $i\in[0,H]$, $p\in S_i$, $a\in A$ and $p'\in S'$,
  \begin{align*}
  P'(p,a)(p')=&\begin{cases}
  P(\last(p),a)(\last(p')) &\text{if }i<H\text{ and }\exists s'\in S, p'=p\cdot as'\\
  1 &\text{if }i=H\text{ and }p'=p\\
  0 &\text{otherwise,}
\end{cases}\\
  R'(p,a)=&\begin{cases}
  R(\last(p),a) &\text{if }i<H\\
  0 &\text{otherwise.}
\end{cases}\\
  R'_T(p)=&R_T(\last(p))%
\end{align*}
\end{definition}

\begin{proof}[Proof of Lemma~\ref{lm:unfolding}]
  Let us prove that for all $i\in[0,H]$ and all $p\in S_i$,
  \begin{itemize}
    \item $\Val^{H-i}_{M}(\last(p))=\Val^{H-i}_{T(M,s_0,H)}(p)$, and
    \item $\opt^{H-i}_{M}(\last(p))=\opt^{H-i}_{T(M,s_0,H)}(p)$.
  \end{itemize}
  We prove the first statement by induction on $H-i$. For $H-i=0$, for all $p\in S_i$,
  $\Val^{H-i}_M(\last(p))=\Val^{H-i}_{T(M,s_0,H)}(p)=R_T(\last(p))$.
  Assume the statement is true for $H-i=k$, so that for all $p\in S_{H-k}$,
  $\Val^k_M(\last(p))=\Val^k_{T(M,s_0,H)}(p)$.
  Then for all $p \in S_{H-k-1}$, we have for all $a\in A$ and $s\in Supp(P(\last(p),a))$,
  $\Val^k_M(s)=\Val^k_{T(M,s_0,H)}(p\cdot a s)$.
  It follows that
  \begin{align*}
    \Val^{k+1}_M(\last(p))&=\max_{a\in A}(R(\last(p),a)+\sum_sP(\last(p),a)\Val^k_M(s))\\
    &=\max_{a\in A}(R(\last(p),a)+\sum_sP(\last(p),a)\Val^k_{T(M,s_0,H)}(p\cdot as))\\
    &=\Val^{k+1}_{T(M,s_0,H)}(p)\,.
  \end{align*}
  From $\Val^{H-i}_{M}(\last(p))=\Val^{H-i}_{T(M,s_0,H)}(p)$ and $\opt^H_M(\last(p))=\argmax_{a\in A}(R(\last(p),a)+\sum_sP(\last(p),a)\Val^{H-1}_M(s))$ we derive $\opt^{H-i}_{M}(\last(p))=\opt^{H-i}_{T(M,s_0,H)}(p)$.
\end{proof}

\section{UCB}\label{app:UCB}

Let $\overline X_{a,n}=\frac{1}{n}\sum_{t=1}^{n}X_{a,t}$ denote the average of
the first $n$ plays of action $a$.
Let $\mu_{a,n}=\E[\overline X_{a,n}]$.
We assume that these expected means eventually converge,
and let $\mu_a=\lim_{n\to\infty}\mu_{a,n}$.

\begin{definition}[Drift conditions]\label{def:drift-app}
\item For all $a\in A$, the sequence $(\mu_{a,n})_{n\geq 1}$ converges to some value $\mu_a$
\item There exists a constant $C_p>0$ and an integer $N_p$ such that for $n\geq N_p$ and any $\delta>0$, $\Delta_n(\delta)=C_p\sqrt{n\ln(1/\delta)}$, the following bounds hold:
	$$\Pr\Big[n\overline X_{a,n}\geq n \mu_{a,n} +\Delta_n(\delta)\Big]\leq \delta\,,$$
	$$\Pr\Big[n\overline X_{a,n}\leq n \mu_{a,n} -\Delta_n(\delta)\Big]\leq \delta\,.$$
\end{definition}

We define $\delta_{a,n}=\mu_{a,n}-\mu_a$.
Then, $\mu^*$, $\mu^*_n$, $\delta^*_n$ are defined as $\mu_j$, $\mu_{j,n}$,
$\delta_{j,n}$ where $j$ is the optimal action.\footnote{It is assumed for simplicity
that a single action $a$ is optimal, \textit{i.e.}~maximises $\E[X_{a,n}]$
for $n$ large enough.}
Moreover, let $\Delta_a=\mu^*-\mu_a$.

As $\delta_{a,n}$ converges to $0$ by assumption, for all $\epsilon>0$
there exists $N_0(\epsilon)\in\N$, such that for $t>N_0(\epsilon)$,
then $2|\delta_{a,t}|\leq \epsilon\Delta_a$ and $2|\delta^*_{t}|\leq \epsilon\Delta_a$
for all all suboptimal actions $a\in A$.

The authors start by bounding the number of time a suboptimal action is played:
 \begin{theorem}[{\cite[Theorem~1]{DBLP:conf/ecml/KocsisS06}}]\label{thm:suboptimal-upper-bound}
  Consider UCB1 applied to a non-stationary bandit problem
  with $c_{t,s}=2C_p\sqrt{\frac{\ln t}{s}}$. Fix $\epsilon>0$.
  Let $T_a(n)$ denote number of times action $a$ has been played at time $n$.
  Then under the drift conditions, there exists $N_p$ such that
  for all suboptimal actions $a\in A$,
	$$\mathbb E[T_a(n)]\leq \frac{16C^2_p\ln n}{(1-\epsilon)^2\Delta_a^2}+N_0(\epsilon)+N_p+1+\frac{\pi^2}{3}\,.$$
\end{theorem}

Let $\overline X_{n}=\sum_{a\in A}\frac{T_a(n)}{n}\bar{X}_{a,T_a(n)}$ denote
the global average of payoffs received up to time $n$

Then, one can bound the difference between $\mu^*$ and $\overline X_{n}$:
\begin{theorem}[{\cite[Theorem~2]{DBLP:conf/ecml/KocsisS06}}]\label{thm:bias-bound}
	Under the drift conditions of Definition~\ref{def:drift-app}, it holds that
  $$|\E[\overline X_n]-\mu^*|\leq |\delta^*_n|+O\left(\frac{|A|(C_p^2\ln n+N_0(1/2))}{n}\right)\,.$$
\end{theorem}

The following theorem shows that the number of times an action is played can be lower bounded:
\begin{theorem}[{\cite[Theorem~3]{DBLP:conf/ecml/KocsisS06}}]\label{thm:lower-bound}
	Under the drift conditions of Definition~\ref{def:drift-app}, there exists some positive constant $\rho$ such that after $n$ iterations for all action $a$, $T_a(n)\geq \lceil\rho\ln(n)\rceil$.
\end{theorem}

Then, the authors also prove a tail inequality similar to the one described
in the drift conditions, but on the random variable $\overline X_{n}$
instead of $\overline X_{a,n}$.
\begin{theorem}[{\cite[Theorem~4]{DBLP:conf/ecml/KocsisS06}}]\label{thm:drift-conditions}
	Fix an arbitrary $\delta>0$ and let $\Delta_n=9\sqrt{2n\ln(2/\delta)}$.
  Let $n_0$ be such that $\sqrt{n_0}\geq O(|A|(C^2_p\ln n_0+N_0(1/2))$.
  Then under the drift conditions, for any $n\geq n_0$, the following holds true:
	$$\Pr[n\overline X_n\geq n\E[\overline{X}_n] +\Delta_n(\delta)]\leq \delta$$
	$$\Pr[n\overline X_n\leq n\E[\overline{X}_n]-\Delta_n(\delta)]\leq \delta$$
\end{theorem}

Finally, the authors argue that the probability of making the wrong decision
(choosing a suboptimal action) converges to $0$ as the number of plays grows:
\begin{theorem}[{\cite[Theorem~5]{DBLP:conf/ecml/KocsisS06}}]\label{thm:failure-prob}
	Let $I_t$ be the action chosen at time $t$, and let $a^*$ be the optimal action.
	Then $\lim_{t \to\infty}Pr(I_t\neq a^*)=0$.
\end{theorem}

\section{MCTS with Simulation}

After $n$ iterations of MCTS, we have $\total_n(p)=\sum_{i\mid I_i(p)\leq n}\reward_{I_i(p)}(p)$ and $\total(p,a)=\sum_{i\mid I_i(p,a)\leq n}\reward_{I_i(p,a)}(p,a)$.

We use the following observations, derived from the structure of the MCTS algorithm. For all nodes $p$
in the search tree, after $n$ iterations, we have:
\begin{align*}
\total_n(p)&=\reward_{\I_1(p)}(p)+\sum_{a\in A}\total_n(p,a)\\
\total_n(p,a)&=\sum_{s\in\Supp(P(\last(p),a)}\total_n(p\cdot as)+R(\last(p),a)\cdot \numsamples_n(p,a)\\
\mctsvalue_n(p)&=\frac{\total_n(p)}{\numsamples_n(p)}\\
\numsamples_n(p)&= 1+\sum_a\numsamples_n(p,a)\\
\numsamples_n(p,a)&= \sum_s\numsamples_n(p\cdot as)
\end{align*}

In the following proof we will abuse notations slightly and conflate the variables and counters used in MCTS
with their associated random variables, \textit{e.g.}~we write $\E[\mctsvalue_n(s_0)]$ instead of
$\E[V_n(s_0)]$ with $V_n(s_0)$ a random variable that represents the value $\mctsvalue_n(s_0)$.
\begin{proof}[Proof of Lemma~\ref{lm:MAB}] %
We use the following inequality (Chernoff-Hoeffding inequality){\cite[Theorem~2]{doi:10.1080/01621459.1963.10500830}} throughout the proof:

  Let $X_1,X_2,\dots X_n$ be independent random variables in $[0,1]$. Let $S_n=\sum_nX_i$. Then for all $a>0$, $\Pr\Big[S_n\geq \E[S_n]+t\Big]\leq \exp{\left(-\frac{2t^2}{n}\right)}$ and $\Pr\Big[S_n\leq \E[S_n]-t\Big]\leq \exp{\left(-\frac{2t^2}{n}\right)}$.

We need to show that the following conditions hold:
  \begin{enumerate}
  	\item $\lim_{\numsamples_n(p)\to\infty}\E[\mctsvalue_n(p,a)]$ exists for all $a$.
  	\item\label{it:tail} There exists a constant $C_p>0$ such that for $\numsamples_n(p,a)$ big enough and any $\delta>0$, $\Delta_{\numsamples_n(p,a)}(\delta)=C_p\sqrt{\numsamples_n(p,a)\ln(1/\delta)}$, the following bounds hold:

  	$$\Pr\Big[\total_n(p,a)\geq \E[\total_n(p,a)] +\Delta_{\numsamples_n(p,a)}(\delta)\Big]\leq \delta$$
  	$$\Pr\Big[\total_n(p,a)\leq \E[\total_n(p,a)] -\Delta_{\numsamples_n(p,a)}(\delta)\Big]\leq \delta$$
  \end{enumerate}
 We show it by induction on $H-|p|$.

  For $|P|=H-1$: $\reward_i(p,a)$ follows a stationary distribution: $\reward_i(p,a)=R(\last(p),a)+R_T(s)$ with probability $P(\last(p),a)(s)$. Thus
\begin{align*}
\E[\total_n(p,a)]&=\E\left[\sum_{i\mid I_i(p,a)\leq n}\reward_{I_i(p,a)}(p,a)\right]\\
&=\numsamples_n(p,a)\left(\sum_sR_T(s)P(\last(p),a)(s)+R(\last(p),a)\right)\,.
\end{align*}
Thus $\E[\mctsvalue_n(p,a)]=\sum_sR_T(s)P(\last(p),a)(s)+R(\last(p),a)$.

From the Chernoff-Hoeffding inequality,
    $$\Pr\left[\sum_{i\mid I_i(p,a)\leq n}\reward_{I_i(p,a)}(p,a)\geq
    \E\left[\sum_{i\mid I_i(p,a)\leq n}\reward_{I_i(p,a)}(p,a)\right] +
\sqrt{\frac{\numsamples_n(p,a)}{2}\ln{\frac{1}{\delta}}}
    \right]
    \leq \delta \,,$$
    $$\Pr\left[\sum_{i\mid I_i(p,a)\leq n}\reward_{I_i(p,a)}(p,a)\leq
    \E\left[\sum_{i\mid I_i(p,a)\leq n}\reward_{I_i(p,a)}(p,a)\right] -
\sqrt{\frac{\numsamples_n(p,a)}{2}\ln{\frac{1}{\delta}}}
    \right]
    \leq \delta\,.$$

    Therefore, condition~\ref{it:tail} also holds with $C_p = \frac{1}{\sqrt{2}}$.

  Assume that the conditions are true for all $p\cdot as$. Then, from Theorem~\ref{thm:bias-bound} we get:
  \begin{align*}
    &\left|\E\left[\frac{\sum_{a'}\total_n(pas,a')}{\sum_{a'}\numsamples_n(pas,a')}\right]-
    \lim_{\numsamples_n(p\cdot as)\to\infty}\E\left[\mctsvalue_n(p\cdot as,a^*)\right]\right|\\
    &\leq \left|\E\left[\mctsvalue_n(p\cdot as,a^*)\right]-
    \lim_{\numsamples_n(p\cdot as)\to\infty}\E\left[\mctsvalue_n(p\cdot as,a^*)\right]\right|+\mathcal O\left(\frac{\ln (\numsamples_n(p\cdot as)-1)}{\numsamples_n(p\cdot as)-1}\right)\,,
  \end{align*}
  where $a^*$ is the optimal action from $p\cdot as$. Now,
\begin{align*}
\lim_{\numsamples_n(p)\to\infty}\E[\mctsvalue_n(p\cdot as)]
&=\lim_{\numsamples_n(p)\to\infty}\E\left[\frac{\total_n(p\cdot as)}{\numsamples_n(p\cdot as)}\right]\\
&=\lim_{\numsamples_n(p)\to\infty}\E\left[\frac{\total_n(p\cdot as)-\reward_{\I_1(p)}(p\cdot as)}{\numsamples_n(p\cdot as)-1}\right]\\
&=\lim_{\numsamples_n(p)\to\infty}\E\left[\frac{\sum_{a'}\total_n(p\cdot as,a')}{\sum_{a'}\numsamples_n(p\cdot as,a')}\right]\,.
\end{align*}
Let $\lim_{\numsamples_n(p\cdot as)\to\infty}\E[\mctsvalue_n(p\cdot as,a^*)]$ be denoted by $\mu_{p\cdot as}$ (we know that this limit exists by the induction hypothesis).
From Theorem~\ref{thm:lower-bound}, we know that $\numsamples_n(p,a)\to \infty$ for all $a$ when $\numsamples_n(p)\to \infty$.
And as for all states $s$, state $s$ is chosen according to distribution $P(p,a)(s)$, $\numsamples_n(p\cdot as)\to \infty$ with probability $1$. Then,
  \begin{align*}
    \lim_{\numsamples_n(p)\to\infty}\E[\mctsvalue_n(p\cdot a)]&=\lim_{\numsamples_n(p)\to\infty}\E\left[\sum_s\mctsvalue_n(p\cdot as)\frac{\numsamples_n(p\cdot as)}{\numsamples_n(p,a)}+R(\last(p),a)\right]\\
    &=R(\last(p),a)+\sum_s\mu_{p\cdot as}\cdot P(\last(p),a)(s)\,.
  \end{align*}
So $\lim_{\numsamples_n(p)\to\infty}\E[\mctsvalue_n(p\cdot a)]$ exists.

  From Theorem~\ref{thm:drift-conditions}, when $\numsamples_n(p\cdot as)$ is big enough,
  for all $\delta>0$,
  $\Pr\Big[\sum_{a'}\total_n(pas,a')\geq \E[\sum_{a'}\total_n(pas,a')]+\Delta_1^s(\delta)\Big]\leq \frac{\delta}{2|S|}$ where $\Delta_1^s(\delta)=9\left(\sqrt{\numsamples_n(p\cdot as)\ln\left(\frac{4|S|}{\delta}\right)}\right)$.

Therefore $\Pr\Big[\total_n(p\cdot as)-\reward_{\I_1(p\cdot as)}(p\cdot as)\geq \E[\total_n(p\cdot as)-\reward_{\I_1(p\cdot as)}(p\cdot as)]+\Delta_1^s(\delta)\Big]\leq \frac{\delta}{2|S|}$.
  Also the random variable associated to $\reward_{\I_1(p\cdot as)}(p\cdot as)$ following a fixed stationary distribution $f(p)$ in $[0,1]$. So from the Chernoff-Hoeffding inequality,
  $\Pr\Big[\reward_{\I_1(p\cdot as)}(p\cdot as)\geq \E[\reward_{\I_1(p\cdot as)}(p\cdot as)]
  +\Delta_2(\delta)\Big]\leq \frac{\delta}{2|S|}$
  where $\Delta_2(\delta) = \frac{1}{\sqrt{2}}\left(\sqrt{\ln\left(\frac{2|S|}{\delta}\right)}\right)$.

  Now using the fact that for $n$ random variables $\{A_i\}_{i\leq n}$ and $n$ random variables $\{B_i\}_{i\leq n}$, $\Pr[\sum_iA_i\geq \sum_i B_i]\leq \sum_i \Pr[A_i\geq B_i]$, we get:
  \begin{align*}
    &\Pr\Big[\total_n(p\cdot as)\geq \E[\total_n(p\cdot as)]+\Delta_1^s(\delta)+\Delta_2(\delta)\Big]\\
    &\leq \Pr\Big[\total_n(p\cdot as)-\reward_{\I_1(p\cdot as)}(p\cdot as)\geq \E[\total_n(p\cdot as)-\reward_{\I_1(p\cdot as)}(p\cdot as)]+\Delta_1^s(\delta)\Big]\\
    &\quad \quad +\Pr\Big[\reward_{\I_1(p\cdot as)}(p\cdot as)\geq \E[\reward_{\I_1(p\cdot as)}(p\cdot as)]+\Delta_2(\delta)\Big]\leq \frac{\delta}{|S|}\,.
  \end{align*}
  As $\numsamples_n(p,a)\E[R(\last(p),a)]=\E[\numsamples_n(p,a)\cdot R(\last(p),a)]$,
  \begin{align*}
    &\Pr\left[\total_n(p,a)\geq \E[\total_n(p,a)]+\sum_s(\Delta_1^s(\delta)+\Delta_2(\delta))\right]\\
    &\leq \sum_s \Pr\Big[\total_n(p\cdot as)\geq \E[\total_n(p\cdot as)]+(\Delta_1^s(\delta)+\Delta_2(\delta))\Big]\leq \delta\,.
  \end{align*}

  Similarly, when $\numsamples_n(p\cdot as)$ is big enough, for all $\delta>0$ it holds that
  $\Pr\Big[\total_n(p\cdot as)-\reward_{\I_1(p\cdot as)}(p\cdot as)
  \leq \E[\total_n(p\cdot as)-\reward_{\I_1(p\cdot as)}(p\cdot as)]-\Delta_1^s(\delta)\Big]
  \leq \frac{\delta}{2|S|}$ and
  $\Pr\Big[\reward_{\I_1(p\cdot as)}(p\cdot as)\leq \E[\reward_{\I_1(p\cdot as)}(p\cdot as)]-\Delta_2(\delta)\Big]\leq \frac{\delta}{2|S|}$.

  Thus $\Pr\Big[\total_n(p,a)\leq \E[\total_n(p,a)]-\sum_s\left(\Delta_1^s(\delta)+\Delta_2(\delta)\right)\Big]\leq \delta$.

  As $\numsamples_n(p\cdot as)\leq \numsamples_n(p,a)$, there exists $C\in\N$ such that for $\numsamples_n(p\cdot as)$ big enough and for all $\delta>0$:
  \begin{align*}
    \sum_s(\Delta_1^s(\delta)+\Delta_2(\delta))&\leq
    C\sum_s\sqrt{\numsamples_n(p\cdot as)\ln\left(\frac{1}{\delta}\right)}\\
    &\leq C\sum_s\sqrt{\numsamples_n(p,a)\ln\left(\frac{1}{\delta}\right)}\\
    &\leq C |S| \sqrt{\numsamples_n(p,a)\ln\left(\frac{1}{\delta}\right)}
  \end{align*}
  So, there is a constant $C_p$ such that for $\numsamples_n(p,a)$ big enough and any $\delta>0$,
  it holds that $\Delta_{\numsamples_n(p,a)}(\delta)=C_p\sqrt{{\numsamples_n(p,a)}\ln(1/\delta)}\geq \sum_s(\Delta_1^s(\delta)+\Delta_2(\delta))$. Therefore, the following bound hold:
  $\Pr\Big[\total_n(p,a)\geq \E[\total_n(p,a)] +\Delta_{\numsamples_n(p,a)}(\delta)\Big]$
  is upper bounded by $\Pr\Big[\total_n(p,a)\geq \E[\total_n(p,a)]+\sum_s(\Delta_1^s(\delta)+\Delta_2(\delta))\Big]$.
  It follows that $\Pr\Big[\total_n(p,a)\geq \E[\total_n(p,a)] +\Delta_{\numsamples_n(p,a)}(\delta)\Big]\leq \delta$.
  Similarly, $\Pr\Big[\total_n(p,a)\leq \E[\total_n(p,a)] -\Delta_{\numsamples_n(p,a)}(\delta)\Big]$
  is upper bounded by $\Pr\Big[\total_n(p,a)\leq \E[\total_n(p,a)]-\sum_s(\Delta_1^s(\delta)+\Delta_2(\delta))\Big]$. It follows that $\Pr\Big[\total_n(p,a)\leq \E[\total_n(p,a)] -\Delta_{\numsamples_n(p,a)}(\delta)\Big]\leq \delta$.

  This proves that for any $p$, the sequences $(x_{a,t})_{t\geq 1}$ associated with
  $\reward_{\I_t(p,a)}(p)$ satisfy the drift conditions.
\end{proof}

\end{document}